\documentclass[acmsmall,table]{acmart}
\setcopyright{cc}
\setcctype{by}
\acmDOI{10.1145/3776708}
\acmYear{2026}
\acmJournal{PACMPL}
\acmVolume{10}
\acmNumber{POPL}
\acmArticle{66}
\acmMonth{1}
\received{2025-07-10}
\received[accepted]{2025-11-06}

\title{\name: A Programmable Framework for Semantically Constraining the Output of Language Models}

\usepackage{graphicx}
\usepackage{amsmath}
\usepackage{arydshln}
\usepackage{amsthm}
\usepackage{wrapfig}
\usepackage{tikz}
\usepackage{tikz-cd}
\usepackage{varwidth}
\usepackage{adjustbox}
\usetikzlibrary{trees}
\usetikzlibrary{shapes.geometric, calc, arrows.meta, positioning}
\usepackage{cleveref}
\usepackage{enumerate}
\usepackage{algorithm}
\usepackage[noend]{algpseudocode}
\usepackage{ifthen}
\usepackage{xspace}
\usepackage{listings}
\usepackage{mathpartir}
\usepackage{multirow}
\usepackage[normalem]{ulem}
\usepackage[dvipsnames]{xcolor}

\usepackage{leftindex}
\usepackage{caption}
\usepackage{subcaption}
\usepackage{stmaryrd}
\usepackage{tcolorbox}
\usepackage{makecell}
\usepackage{listings}
\usepackage{droidsansmono}

\lstdefinelanguage{haskell}{
  sensitive=true,
  morecomment=[l]--,                % line comments
  morecomment=[s]{\{-}{-\}},        % block comments
  morestring=[b]",                  % strings in quotes
  alsoletter={:},                   % allow :: as operator
  alsoletter={-},                   % allow ->, =>, etc.
  morekeywords=[1]{
    match, with,if,then,else,
    let,in,data,type,newtype,class,instance,
    where,do,import,qualified,hiding,as,
    deriving,default,infix,infixl,infixr,
    module,
  },
  otherkeywords={
    ::, =, ->, =>, <-, =, |, &, \\, -- % add operators if you want them highlighted
  }
}

\algnewcommand{\LeftComment}[1]{\Statex \(\triangleright\) #1}

\newcommand{\Omit}[1]{}
\newcommand{\changed}[1]{{{{#1}}}}

\theoremstyle{definition}

\newcommand{\emphbf}[1]{\textbf{\emph{#1}}}
\newcommand{\code}[1]{\lstinline{#1}}
\newcommand{\mcode}[1]{\text{\texttt{#1}}}

\newcommand{\abs}[1]{|#1|}
\newcommand{\rone}{(\emph{i})~}
\newcommand{\rtwo}{(\emph{ii})~}
\newcommand{\rthree}{(\emph{iii})~}

\def\bool{\mathrm{bool}}
\def\int{\mathrm{int}}

\def\false{\mathrm{false}}

\newcommand{\llangle}{\langle\!\langle}
\newcommand{\rrangle}{\rangle\!\rangle}

\newcommand{\semlexerstateprog}[1]{\llangle{#1}\rrangle^{\mathsf{lexerstate}}}

\newcommand{\Tau}{\mathbf{T}}

\newcommand{\name}{\textsc{ChopChop}\xspace}
\def\toolname{\name}

\def\token{\tau}
\def\tokens{\Tau}

\def\stringConstraint{\Phi}
\def\ASTConstraint{\phi}
\def\st{^*}
\def\alphabet{\Sigma}
\def\AST{\mathsf{AST}}

\def\prefixSpace{\mathsf{PrefixSpace}}

\def\llm{\mathsf{lm}}
\def\wl{\mathsf{worklist}}

\def\lexset{\mathsf{partial\_lex}}

\def\realchecker{\mathsf{check\_realizable}}
\def\monitor{\mathsf{monitor}}
\def\realizable{\mathsf{realizable}}
\def\completable{\mathsf{completable}}
\def\dequeue{\mathsf{dequeue}}
\def\enqueue{\mathsf{enqueue}}
\def\endtok{\ensuremath{\mathsf{END}}}
\def\parse{\mathsf{parse}}

\def\lexerstate{\tilde{L}}
\def\growlexerstate{\mathsf{compute\_lexer\_state}}
\def\rmvpointers{\mathsf{remove\_annotations}}
\def\rmvignores{\mathsf{remove\_ignorable\_tokens}}

\def\lex{\mathsf{lex}}
\def\extendlexes{\mathsf{extend\_lexer\_state}}
\def\munch{\mathsf{remove\_nonmaximal\_munch\_lexes}}

\def\pltoken{\alpha}

\def\typerestrict{\code{type_prune}\xspace}

\author{Shaan Nagy}
\orcid{https://orcid.org/0000-0001-8015-5421}
\authornote{Shaan Nagy and Timothy Zhou contributed equally to this paper.}
\affiliation{
\institution{University of California San Diego}
\country{USA}
}
\author{Timothy Zhou}
\authornotemark[1]
\orcid{0000-0002-5262-0995}
\affiliation{
\institution{University of California San Diego}
\country{USA}
}
\author{Nadia Polikarpova}
\orcid{https://orcid.org0000-0001-5571-173X}
\affiliation{
\institution{University of California San Diego}
\country{USA}
}
\author{Loris D'Antoni}
\orcid{0000-0001-9625-4037}
\affiliation{
\institution{University of California San Diego}
\country{USA}
}

\begin{abstract}
Language models (LMs) can generate code but cannot guarantee its correctness---often producing outputs that violate type safety, program invariants, or other semantic properties. Constrained decoding offers a solution by restricting generation to only produce programs that satisfy user-defined properties. However, existing methods are either limited to syntactic constraints or rely on brittle, ad hoc encodings of semantic properties over token sequences rather than program structure.

We present \name, the first programmable framework for constraining the output of LMs with respect to semantic properties. \name introduces a principled way to construct constrained decoders based on analyzing the space of programs a prefix represents. It formulates this analysis as a realizability problem which is solved via coinduction, connecting token-level generation with structural reasoning over programs. We demonstrate \name's generality by using it to enforce (1) equivalence to a reference program and (2) type safety. Across a range of models and tasks, \name improves success rates while maintaining practical decoding latency.
\end{abstract}

\begin{document}




\maketitle

\section{Introduction}
\label{sec:introduction}
Language models (LMs) have fundamentally transformed how we interact with code---generating functions, completing boilerplate, and even suggesting entire programs.
Yet despite their success, LMs offer no guarantees of correctness: they 
often produce code that looks plausible but violates critical syntactic or semantic properties.
This limitation arises because LMs generate text in a purely probabilistic fashion---at each step, they sample the next token from a distribution conditioned on the sequence they have produced so far.

\textit{Constrained decoding} addresses this problem by enforcing constraints on the sequence an LM generates~\cite{wang2023grammar,willard2023efficient,geng2023grammarconstrained,monitors,typespaper}: a constrained decoder modifies the sampling process to forbid selecting next tokens that would violate the constraint.
Concretely, a constrained decoder performs a \textit{completability analysis} at each step to determine whether the partially generated sequence can still be extended into a completion that satisfies the constraints.
If not, the LM backtracks and tries another token.
By construction, any output a constrained decoding algorithm returns satisfies the user-provided constraints.

Existing constrained-decoding techniques are limited in either \emph{expressivity} (the complexity of the constraints they can enforce) or \emph{programmability} (the ease with which users can specify constraints).

\textbf{Programmable but not Expressive.} Most of the constrained-decoding literature is interested in the problem of forcing the LM to only produce sequences that are accepted by a user-given context-free grammar (CFG)~\cite{greatgramma,gad,syncode,llguidance}.
While CFGs provide a programmable interface for specifying \textit{syntactic} constraints, many desirable \emph{semantic} properties---such as type safety or adherence to invariants---cannot be expressed as context-free languages.

\textbf{Expressive but not Programmable.} 
% A few techniques can enforce certain kinds semantic constraints, but are typically hard to instantiate with rich semantic constraints~\cite{typespaper,monitors}.
%
\citet{monitors} proposed an automaton-based framework (which they call monitors) to connect string-based static analyses (e.g., those offered by a Language Server Protocol) to a constrained decoder.
\citet{typespaper} implement a monitor to enforce well-typedness of TypeScript code.
% extend monitors to the problem of enforcing well-typedness.
%
Although monitors are expressive, they operate on the string representation of a program.
Because semantic constraints are usually defined over abstract syntax trees, it is nontrivial (and often unnatural) to translate semantic constraints into a corresponding completability analysis over strings.
For example, while an expression like \code{1 + myString} would likely be considered ill-typed as a \textit{complete} program, as an \textit{incomplete} program it may still admit a valid continuation (e.g., \code{1 + myString.length}).
Furthermore, because monitors are not defined over abstract syntax, they must manually handle syntax-level concerns---lexing, precedence, left recursion, and so on---failing to separate syntactic and semantic constraints.
% provide a clean separation between syntactic and semantic constraints.
%
% In short, b
By operating over concrete syntax, monitors sacrifice \emph{programmability}
and offer no clear path for enforcing program equivalence or other complex semantic invariants.

In this paper, we ask:
\begin{center}
% \emph{Can we design a principled, programmable framework for constrained decoding that enforces deep semantic properties---over programs, not just token sequences?}
\emph{Can we design a framework for constrained decoding that is both \textbf{expressive} and \textbf{programmable}?}
\end{center}

% Achieving this goal introduces \snchanged{a} major challenge\snchanged{:}

% \paragraph{Bridging the syntax-semantics gap.} Formal methods reason about semantic properties over abstract syntax trees (ASTs) or \snchanged{other} intermediate representations, while language models generate concrete syntax one token at a time. Enforcing semantic constraints during decoding requires translating between the evolving token prefix and the corresponding space of possible ASTs the prefix might generate.

% First, formal methods offer powerful techniques for verifying semantic properties---ranging from type soundness to program equivalence---but these techniques reason over abstract syntax trees (ASTs) or term structures, not raw token sequences. In contrast, LMs emit concrete syntax one token at a time. This gap between token-level generation and abstract reasoning makes it hard to directly enforce semantic constraints during decoding.

% \loris{revise post-sec-2-3-4}
% Second, traditional analyses are designed for complete programs. They determine whether a fully formed program satisfies a property, not whether a partial token prefix could still lead to a valid program. Rewriting such analyses to work incrementally at the token level is tedious and fundamentally limits expressiveness.

\paragraph{Our approach.}
We present \name, a framework that enables programmable constrained decoding.
The key novelty of \name is that it allows one to constrain the abstract syntax trees corresponding to the output sequences produced by an LM.

Conceptually, given a semantic constraint over abstract syntax trees, \name has to decide whether a sequence of tokens in the concrete syntax (i.e., the output of an LM) can be completed so that parsing it will produce an abstract syntax tree that satisfies the given constraint.
Given a sequence of tokens, \name solves the above problem in two steps: 
\rone it computes a representation of 
the set of ASTs (a \textit{program space}) producible by completing the token sequence with some suffix string---and
\rtwo it checks whether any AST in this space satisfies the semantic constraint.

More concretely, to use \name, users provide two inputs:
\rone a \emph{grammar} defining the concrete syntax of the language and how abstract syntax trees can be extracted from sequences in the concrete syntax, and
\rtwo a set of \emph{semantic pruners} which describe how to prune constraint-violating ASTs from a program space.
Decoupling the semantic constraint from concrete program syntax separates concerns, allowing for an interface that is easy to program.

The key programming-language insight that enables \name's pipeline to work despite having to deal with infinite sets of ASTs is the following: program spaces can be represented as \emph{regular codata} and semantic pruners can be implemented via \emph{corecursion} in the style of~\cite{cocaml}.

% In the program synthesis literature, this second step is known as \textit{realizability} and has been well-studied ~\cite{unrealizability2019,unrealizability2020}.

%

\paragraph{Applications.}
We demonstrate the generality of \name through two diverse applications:
\begin{itemize}
    \item \textbf{Program equivalence-guided decoding:} We constrain a language model to generate programs equivalent, modulo term rewriting, to a reference program.
    Here, the semantic pruner uses e-graphs~\cite{egg} as program spaces to efficiently represent all programs that are equivalent to the reference program.

    \item \textbf{Type-safe decoding:} We constrain a language model to emit only well-typed programs in a subset of TypeScript.
    Here, the semantic pruner lifts type checking to operate over co-inductively defined sets of abstract syntax trees---i.e., the pruner constrains away possible program completions that will not typecheck.
\end{itemize}

For these domains, \name produces programs that satisfy given semantic constraints more often than both unconstrained language models and approaches that only constrain the program syntax.

\paragraph{Contributions.}
\begin{itemize}
    \item The first \textit{programmable framework for semantic constrained decoding} (\Cref{sec:overview}).
    \item A \textit{formal connection between semantic constrained decoding and realizability} in program synthesis, showing that enforcing semantic properties during generation reduces to checking whether the set of possible completions contains a valid program (\Cref{sec:sem-constraint}).
    % \item \snchanged{A method for \textit{solving realizability problems via regular corecursion} by representing spaces of ASTs as regular codata and pruning them with corecursive functions (\Cref{sec:corecursion}).}
    \item A \textit{novel constrained decoding algorithm based on corecursion}; the algorithm incrementally builds and analyzes a coinductive representation of the space of programs obtainable from a prefix (\Cref{sec:corecursion}).
    \item An \textit{implementation of our framework} in a tool, \name, which we instantiate
    in two challenging domains---equivalence-guided decoding and type-safe generation---demonstrating the generality of our approach (\Cref{sec:case-studies}).
    \item An \textit{evaluation} of \name on the two domains---\name produces programs that satisfy given semantic constraints more often than both unconstrained and purely syntax-constrained baselines while incurring a small computational overhead (\Cref{sec:eval}).
\end{itemize}

\Cref{sec:related_work} discusses related work, and \Cref{sec:conclusion} concludes.
\section{Overview of \name} \label{sec:overview}
We illustrate our approach with the following simple prompt one might task an LM with:
\begin{align*}
\textit{Generate a sum of odd integers whose total is even.}
\end{align*}

\begin{wrapfigure}{r}{0.37\linewidth}
    \vspace{-6.75mm}
    \centering
\begin{minipage}{\linewidth}
        \small    
\[
\begin{array}{rcl@{\hspace{1em}}l}
\text{\textsc{int}} &::=& {[0\text{-}9]+} & \\[2pt]
E &::=& \textsc{int} & \{ \text{Lit \$1} \} \\[2pt]
  &\mid& E~\texttt{"+"}~\textsc{int} & \{ \text{Sum \$1 \$3} \}
\end{array}
\]    
\end{minipage}
    \vspace{-2mm}
    \caption{
        A left-recursive grammar for the language of integer sums.
        Right-hand annotations specify parse tree-to-AST translation (e.g., $\{\text{Sum \$1 \$3}\}$ applies the Sum constructor to the first and third arguments of the second production).
        % \$i refers to the $i$-th symbol (1-indexed) in a production (e.g., the Sum constructor is applied to the LHS and RHS of the +).
        % \snchanged{\toolname supports} left-recursive grammars.
    }
    \vspace{-14mm}
    \label{fig:attribute_grammar}
\end{wrapfigure}
For example, the expression \code{5+7} is a valid solution, since 12 is even and both 5 and 7 are odd.
In general, there is no \textit{guarantee} that an LM will produce a valid response---it could generate a sum with even numbers (e.g., \code{2+2}), an odd total (e.g., \code{1+1+1}), or totally nonsensical output (e.g., \code{banana}).

A user of \toolname can enforce the constraints mentioned in the prompt by providing two inputs:
\begin{enumerate}
    \item 
    \parbox[t]{0.574\linewidth}{
    A \textit{grammar} (\Cref{fig:attribute_grammar}) defining the syntax of the
    language, annotated with rules for translating parse trees to ASTs.}
    \vspace{1mm}
    \item A set of \emph{semantic pruners} (\Cref{fig:example_checker_odds_only,fig:example_checker_even_sum}) each representing a constraint on ASTs.
    At a high-level, a pruner is a function that takes in a representation of a space of programs and returns (an approximation of) the subspace of constraint-satisfying programs.
\end{enumerate}

Given these inputs, \toolname~guarantees that any generated program is syntactically valid and satisfies the provided semantic constraints---in our case, the output being a sum of odd integers with an even total.\footnote{While the constraints being enforced in this example are task-specific, in practice users would write pruners for general constraints that apply across all tasks in an application domain (as in \Cref{sec:case-studies}).} 
Crucially, \toolname enforces constraints \textit{during generation}, rather than waiting for the LM to produce a complete program and then retrying from scratch upon failure.

\subsection{Semantic Constrained Decoding as Realizability}
Returning to our example, suppose the LM has already generated the prefix \code{1+2}, and the LM's top two choices for the next token are \code{+} and \code{2}.
\toolname disallows \code{+} since any program starting with \code{1+2+} will contain an even integer.
However, it allows the next token \code{2} as it can still lead to a valid completion (for example, \code{1+221}).

To rule out the token \code{+}, \toolname must determine that no completion of the sequence \code{1+2+} parses to an AST satisfying our example constraints.
We say that \toolname solves \emph{semantic constrained decoding} because it prunes the LM's choices based on semantic constraints over ASTs (as opposed to \textit{syntactic constrained decoding} which enforces shallow syntactic properties of the token stream).
It is challenging to prune the LM's choices directly because constraints are expressed at the \textit{AST} level, while \code{1+2+} is an incomplete \textit{string} prefix.

Our first key insight is that \textbf{a prefix represents the space of programs corresponding to its completions.}
Instead of attempting to reason about AST-level properties on \textit{string} completions of \code{1+2}, \toolname reasons about the \textit{space of ASTs} that this prefix represents.
In the synthesis literature, checking satisfiability of constraints over spaces of ASTs is studied under the name \emph{realizability}~\cite{unrealizability2019}.
The next section outlines our realizability-based approach to semantic constrained decoding in more detail.

\subsection{Realizability as Analysis over Regular Codata}\label{sec:overview:realizability}
Next, we show how to formulate and solve the realizability problem to which completability of the prefix \code{1+2} reduces. We split this problem into three goals:
\begin{enumerate}
    \item \textbf{Representation:} How can we finitely describe infinite program spaces?
    \item \textbf{Completion:} Given a string $\omega$, how can we construct the space of all ASTs corresponding to completions of $\omega$?
    \item \textbf{Analysis:} Does there exist a program in the constructed space satisfying the constraint?
\end{enumerate}
Our next key insight is that \textbf{regular codata provides a structured abstraction for reasoning about spaces of programs}. Our solutions to (1)-(3) are built around this idea.
In the rest of the section, we describe our approach following the overview given in \Cref{fig:toolname-flow}.
% TikZ styles
\tikzstyle{box} = [
  rectangle, rounded corners=2pt, fill=white!95,
  minimum height=1.1cm,
  align=left, font=\small
]
\tikzstyle{arrow} = [thick, -{Stealth}]
\tikzstyle{title} = [font=\bfseries\small]

\begin{figure}[t]
\centering
\resizebox{\textwidth}{!}{%
\begin{tikzpicture}[node distance=1.2cm and 1.5cm]

% Top row
\node[box] (prefix) {
  \textbf{Prefix} \\
  \code{1+2}
};

\node[box, right=of prefix] (lexer) {
  \textbf{Lexical Prefixes} \\
  \code{["1","+","2"]} \\
  {\code{["1","+","2[0-9]+"]}}
};

\node[box, right=of lexer] (parser) {
  \textbf{Program Space} \\
  \code{S = Union [} \\
  \quad\ \code{Sum (Lit "1") (Lit "2"),} \\
  \quad\ \code{ Sum (Lit "1") (Lit "2[0-9]+"),} \\
  \quad\ {\code{ Sum S (Lit "[0-9]+"})} \\
  ]
};
% Bottom node: Pruned Space
\node[box, below=of parser] (pruned) {
  \textbf{Pruned Space} \\
  \code{S' = Union [} \\
  \quad\ \code{Empty,} \\
  \quad\ {\code{ Sum (Lit "1") (Lit "2[0-9]*[13579]"),}} \\
  \quad\ {\code{Sum (odd_sum (odds S)) (Lit "[0-9]*[13579]")}} \\
  ]
};

% Nonempty to the left of pruned
\node[box, left=0.5cm of pruned] (nonempty) {
  \code{nonempty S'}
};

% Continue and discard stacked vertically, left of nonempty
\node[box, above=0.4cm of nonempty, xshift=-2cm] (continue) {
  \textbf{Continue Decoding}
};
\node[box, below=0.4cm of nonempty, xshift=-2cm] (discard) {
  \textbf{Discard Prefix}
};

% Arrows with labels
\draw[arrow] (prefix) -- node[title, above] {Lexer} (lexer);
\draw[arrow] (lexer) -- node[title, above] {Parser} (parser);
\draw[arrow] (parser) -- node[title, right, xshift=2pt, align=center] {
  Apply Pruners: \\ \normalfont{\code{even_sum (odds S)}}
} (pruned);
\draw[arrow] (pruned) -- (nonempty);
\draw[arrow] (nonempty) -- node[right] {\code{true}} (continue);
\draw[arrow] (nonempty) -- node[right] {\code{false}} (discard);

\end{tikzpicture}
}
\caption{
Flow of \toolname on input prefix \code{1+2}. The prefix is lexed into possible lexical sequences, parsed into a symbolic program space, semantically pruned, and checked for nonemptiness to determine realizability. If realizble, the prefix may be extended. If unrealizable, the prefix is discarded.
}
\label{fig:toolname-flow}
\end{figure}
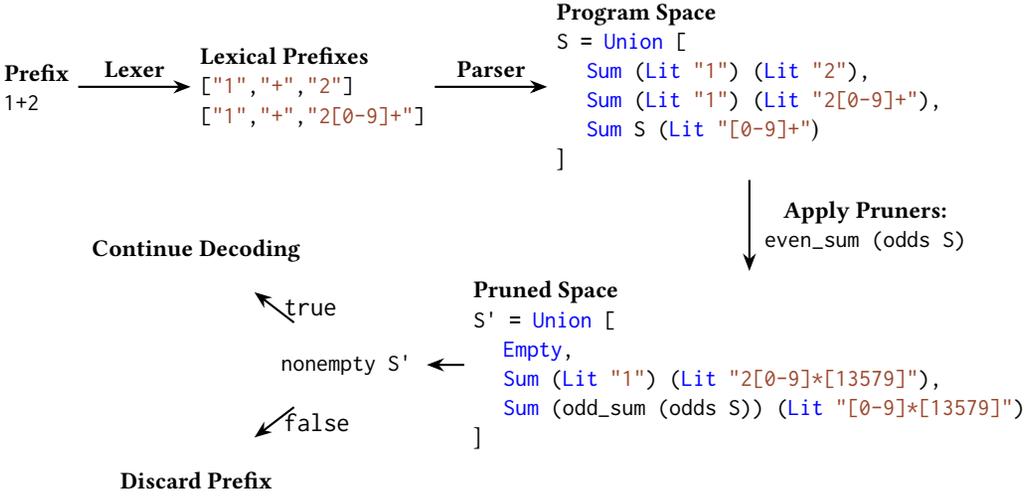

\paragraph{Goal 1: Representing Infinite Program Spaces}
In our example, programs consist solely of integer literals and sums.
To represent an AST for a single program, we can use the following datatype:
\begin{lstlisting}
data Expr = Lit String      -- a numeric literal, e.g., Lit "5"
          | Sum Expr Expr   -- a sum of two expressions
\end{lstlisting}
We would like to lift this definition to represent (possibly infinite) \textit{spaces} of ASTs.
The first step is to extend the signature of \code{Expr} as follows:
\begin{lstlisting}
data ExprSpace = Empty                     
               | Union [ExprSpace]        
               | Lit Regex                
               | Sum ExprSpace ExprSpace 
\end{lstlisting}
An \code{ExprSpace} can be empty,
a union of spaces,
a regular expression that represents a \textit{set} of concrete literals,
or an application of the \code{Sum} constructor lifted to the product of two program spaces.
Readers familiar with \emph{version space algebras}~\cite{VSA2003,FlashMeta2015}
will recognize this representation as a version space, where \code{Union} is the union node and \code{Sum} is a join node.

For spaces to represent infinite sets, we must generalize the shape of terms and allow corecursion.
We represent program spaces as \emph{regular coterms}, i.e., infinite, cyclic terms with a finite description.
For example, the space consisting of \textit{all} sums of integers can be represented using the following corecursive definition:
\begin{lstlisting}
all = 
    let int = Lit "[0-9]+" in 
    Union [int, Sum all int]
\end{lstlisting}
Note that even though \code{all} is self-recursive, it has a \textit{finite} representation in memory as a cyclic term in which the definition of \code{all} contains a back-reference to itself.
Coterms that admit a finite representation (but are potentially infinite) are called regular.

In \name, all program spaces are regular and represented using this finite, cyclic form (\Cref{sec:corecursion}).
We implement and manipulate coterms via a solver inspired by CoCaml~\cite{cocaml}, an extension to OCaml that supports equational reasoning for regular coterms.
This solver can perform computations over program spaces---e.g., applying transformations such as \code{odds} and \code{even_sum} (\Cref{fig:running_example_inputs})---without materializing infinite sets.

\paragraph{Goal 2: Computing Prefix Spaces}
Given a string prefix $\omega$, we need to compute the space of all ASTs that can be produced by parsing any completion of $\omega$.
This construction is done with reference to a user-supplied parser definition that defines a grammar for the language as well as how to convert parse trees to ASTs (\Cref{fig:attribute_grammar}).

As a first step, we lex $\omega$ into a finite set of \emph{lexical prefixes} that encode how $\omega$ can be lexed once its last lexeme has been completed.
For example, the prefix $\omega = \text{1+2}$ has two valid lexical prefixes, representing the different ways $\omega$ could be lexed depending on the next character:
\begin{itemize}
    \item \lstinline{["1","+","2"]}: where the next character is non-numeric (as in \code{1+2+}).
    \item \lstinline{["1","+","2[0-9]+"]}: where the next character continues the numeric literal (as in \code{1+21}).
\end{itemize}

Next, for each lexical prefix we construct a corresponding program space representing its completions.
This construction can be done naturally following the \textit{parsing with derivatives} methodology~\cite{parsingWithDerivatives}.
In our framework, parsers are regular coterms representing the set of possible parse trees producible from their current state%
\footnote{We use the term \emph{lexeme} for programming-language tokens to avoid confusion with LM tokens.}.
Cycles in parsers directly mirror cycles in the underlying context-free grammar.
Critically, parsers support a \emph{derivative} operation that advances the state of a parser as if it read one additional lexeme.\footnote{This approach mirrors the definition of the Brzozowski derivative ~\cite{Brzozowski} of a language $L$ with respect to characer $u$, which is $u^{-1}L = \{w \mid uw \in L\}$.}

For each lexical prefix---e.g., \code{["1","+","2"]}---%
% we initialize the parser from the grammar
we start from the user-provided parser and apply successive derivatives for each lexeme in the prefix;
this results in a parser that accepts exactly those programs that begin with the given lexical prefix.
Finally, we convert each derived parser into a corresponding program space (essentially by discarding the information about concrete syntax) and combine the program spaces from different lexical prefixes using the \code{Union} constructor.
For our example, the prefixes \code{["1","+","2"]} and \code{["1","+","2[0-9]+"]} together would induce the following program space:

\begin{minipage}{\linewidth}
\begin{lstlisting}
S = Union [ Sum (Lit "1") (Lit "2"), 
            Sum (Lit "1") (Lit "2[0-9]+"), 
            Sum S (Lit "[0-9]+") ]
\end{lstlisting}
\end{minipage}

% \begin{align}
% \code{Union (Sum(Sum(1 2) }\termfont{E}\code{)) (Sum(Sum(1 2[0-9]*) }\termfont{E}\code{))}
% \end{align}
% where $\termfont{E} \code{ = Union [1-9][0-9]* (Sum[1-9][0-9]* }\termfont{E}\code{)}$ represents the set of all ASTs derivable from the nonterminal $E$.

% This procedure yields a codata object describing all completions of $\omega$ consistent with the grammar. This object forms the input to the realizability pruner, described next.

% \loris{is there a way to write that exprSpace implements ProgramSpace}

   \begin{figure}[t]
    \centering
    \begin{subfigure}[b]{\textwidth}
        \centering
        \begin{lstlisting}[xleftmargin=0em]
odds :: ExprSpace -> ExprSpace
odds Empty            = Empty
odds (Union children) = Union (map odds children)
odds (Lit regex)      = 
    Lit (regex `intersect` "[0-9]*[13579]") -- only odd
odds (Sum left right) = Sum (odds left) (odds right)
        \end{lstlisting}
        \vspace{-4mm}
    \caption{\texttt{odds} pruner: retains only programs using odd literals.}   \label{fig:example_checker_odds_only}
    \end{subfigure}
    \hfill
    \begin{subfigure}[b]{\textwidth}
        \begin{lstlisting}[xleftmargin=0em]
even_sum :: ExprSpace -> ExprSpace
even_sum Empty            = Empty
even_sum (Union children) = Union (map even_sum children)
even_sum (Lit regex)      = 
    Lit (intersect regex "[0-9]*[02468]") -- only even
even_sum (Sum left right) = 
    Union [ Sum (even_sum left) (even_sum right),  
            Sum (odd_sum left) (odd_sum right) ]

odd_sum :: ExprSpace -> ExprSpace
-- Analogous to even_sum
        \end{lstlisting}
\vspace{-4mm}
        \caption{\code{even_sum} pruner: retains programs whose total evaluates to an even number.
        }
        \label{fig:example_checker_even_sum}
    \end{subfigure}
    \caption{Example semantic pruners.}
    \label{fig:running_example_inputs}
\end{figure}

\paragraph{Goal 3: Analyzing Program Spaces}
The final step towards evaluating realizability is to determine whether there exists a constraint-satisfying program in the prefix space computed in the previous step.
%
% \sn{I don't think $\constraint$ is defined in this section.}
Unfortunately, this problem is undecidable in general, even when the constraint is decidable for individual ASTs~\cite{popl23unrealizabilitylogic,kim2025}.
%
% \sn{Not sure abt this characterization, that pruning is used to solve tractability. But I'm fine if there is a strong preference for this narrative. We could say composability.}
To make the analysis tractable, we therefore turn to the canonical strategy in program synthesis: pruning.
Specifically, users encode their constraints by defining \textit{semantic pruners}, functions that take as input a space of programs and return (an approximation of) the subspace of constraint-satisfying ones.
In \toolname, just as program spaces are modeled as codata, pruners are modeled as corecursive functions.

% \sn{... and ending here.}
Pruners can be composed to enforce conjunctions of constraints,
allowing users to modularly define and reuse semantic constraints across tasks.
For our running example, we supply two pruners (shown in \Cref{fig:running_example_inputs}): \code{odds}, which removes programs with even literals, and \code{even_sum}, which removes programs whose total is odd.
When we apply the two pruners in sequence, we obtain a new regular coterm \code{S'} that represents the subspace of even sums of odd literals:
\begin{lstlisting}
    S' = even_sum (odds S)
\end{lstlisting}
Since \code{S} is cyclic, evaluating \code{S' = even_sum (odds S)} eventually leads back to \code{even_sum (odds S)}. To avoid infinite recursion, our solver tracks previously expanded terms and creates a cycle in \code{S'} when we re-visit a subterm.
For example, the inner term \code{odds S} expands until we reach the state:
\begin{lstlisting}
    odds S = Union [ Sum (Lit "1") Empty,
                     Sum (Lit "1") (Lit "2[0-9]*[13579]"),
                     Sum (odds S) (Lit "[0-9]*[13579]") ]
\end{lstlisting}
%

% \snchanged{Once we have pruned the program space, we must check whether or not it is empty. }
Finally, we must check if there exists a program in the pruned space.
Note that checking emptiness is not the same as checking if the space is exactly the term \code{Empty},
since the presence of cyclic subterms means evaluation may not reduce the space to a normal form. 
\toolname implements a function \code{nonempty :: ExprSpace -> Bool}
that performs a fixpoint computation over the codata representation of its input.
Again, because spaces are \textit{regular} coterms, the fixpoint computation operates over a finite structure despite the fact that a space can represent infinitely many programs.

In our example, \code{nonempty} determines that the first two sub-spaces of the union in \code{S'} are empty
while the third one is non-empty, and hence the union is non-empty.

\paragraph{Putting it All Together}
Together, the four components---lexing, parser derivatives, user-defined pruners, and the nonemptiness check---%
can compute realizability for our running example as:
\begin{lstlisting}
    nonempty (even_sum (odds S))
\end{lstlisting}
% \toolname builds and analyzes prefix spaces within a unified formalism.
%
This pipeline is modular and requires minimal effort from the user:
they define a grammar and supply pruners for each constraint.
\toolname takes care of the rest---automatically enforcing semantic constraints as the LM generates code.

\section{Semantic Constrained Decoding as Realizability Checking}\label{sec:sem-constraint}
In this section, we formalize the problem of semantic constrained decoding and explain how it is equivalent to the problem of checking whether a synthesis problem is realizable---i.e., whether a regular tree grammar accepts any programs that satisfy a given input specification. 
After \Cref{sec:prelim} establishes some basic definitions,
\Cref{sec:constraint-decoding} gives a general formulation of constrained decoding where the constraint is a predicate on strings. 
\Cref{sec:semantic-constrained-decoding} defines the special case of semantic constrained decoding, where the predicate is induced by a constraint on ASTs.
We also provide soundness and completeness theorems.

\subsection{Preliminaries}
\label{sec:prelim}

\paragraph{Strings and Abstract Syntax Trees}

Let $\alphabet$ be a finite alphabet of characters. A \emph{string} is a sequence $\omega \in \alphabet\st$, and the \emph{completions} of $\omega$ are the strings that begin with $\omega$. Let $\AST$ be a set of \emph{abstract syntax trees} (of programs). 
Translation from strings to ASTs is defined by a function $\parse\colon \alphabet^* \rightarrow \AST \cup \{\bot\}$ that produces an AST from a string of characters (or $\bot$ if parsing fails).
\paragraph{Language Models and Decoding}

Language Models (LMs) operate over short character sequences called \emph{tokens}, rather than over individual characters.
We define a \emph{token vocabulary} $\tokens$, where a token $\tau \in \tokens$ is either a finite string of characters from $\alphabet$ or the special token $\endtok$. A token vocabulary $\tokens$ must be large enough to represent any string in $\alphabet\st$ as a sequence of tokens.

We formalize a \emph{language model} $\llm\colon \tokens\st \to \tokens \to [0,1]$ as a function that takes in a \emph{prefix} $\token_1, \ldots, \token_{i-1}\in\tokens\st$
and returns a probability distribution over the next token $\token_i\in\tokens$.
Given a language model, we can iteratively generate a sequence of tokens,
selecting each next token $\token_i$ based on the model's distribution, conditioned on the previously generated tokens $\token_1, \ldots, \token_{i-1}$. This process is called \emph{decoding}.
Decoding can be done greedily (by choosing the most probable token at every step)
or stochastically (e.g., via nucleus sampling \cite{holtzman2020}).
Decoding stops when an $\endtok$-terminated token sequence is produced.
% 
% To account for the many variations, \Cref{fig:decoding} describes decoding in a general form.\loris{how say a sentence about worklist and what lines?}

\subsection{Constrained Decoding}
\label{sec:constraint-decoding}
We begin with a general definition of \textit{constrained decoding} (CD) that is not specific to semantic constraints.
A \emph{string constraint} $\stringConstraint$ is a predicate on $\alphabet\st$ that accepts some set of strings. For example, $\stringConstraint$ may accept the members of some context-free language. We say that a string $\omega \in \alphabet\st$ is \textit{completable} with respect to $\stringConstraint$ if there exists a completion of $\omega$ satisfying $\stringConstraint$:
$$\completable(\omega, \stringConstraint) = \exists \omega' \in \alphabet^*.~\stringConstraint(\omega\omega')$$
A \textit{monitor} for a string constraint $\stringConstraint$ is a function $\monitor_\Phi : \Sigma^* \to \bool$ that checks whether an input string is completable with respect to $\stringConstraint$.
For a given string constraint $\stringConstraint$,
\textit{constrained decoding} maintains the invariant that the currently generated sequence of tokens is completable with respect to a monitor for $\stringConstraint$.
Given a language model $\llm$, a string constraint $\stringConstraint$, and a monitor for $\stringConstraint$, \Cref{alg:constrained-decoding} outlines how constrained decoding operates in its general form.

\renewcommand{\algorithmicindent}{0.7em}
\algdef{C}[IF]{IF}{Elif}[1]{\textbf{elif}\ #1\ \algorithmicthen}

\begin{figure*}
 \begin{minipage}[t]{0.75\textwidth}
 \hspace{-7em}
\begin{algorithm}[H]
\caption{Constrained Decoding}
\begin{algorithmic}[1] \label{alg:constrained-decoding}
\State \textbf{Input:} LM $\llm$, {\color{ForestGreen}string constraint $\stringConstraint$, monitor $\monitor_\stringConstraint$}
\State \textbf{Output:} String $\omega$ {\color{ForestGreen} such that $\stringConstraint(\omega)$ holds}

\State $\wl \gets [(\epsilon, 1)]$
\While{$\wl \neq []$}
    \State $\omega, p \gets \wl.\dequeue()$
    \If{$\omega[-1] = \endtok  {\color{ForestGreen}~\land~\stringConstraint(\omega[:-1])}$}
        \State \Return{$\omega[:-1]$}
    \Elif{$\omega[-1] \neq \endtok  {\color{ForestGreen}~\land~\monitor_\stringConstraint(\omega)}$}
        \For{$\tau \in \tokens$}
            \State $\wl.\enqueue(\omega\tau, p * \llm(\omega, \tau))$
        \EndFor
    \EndIf
\EndWhile
\State \Return{$\bot$}
\end{algorithmic}
       \label{fig:parsingmodules}
\end{algorithm}
\end{minipage}

\caption{Constrained decoding maintains a worklist of pairs of form $(\omega, p)$, where $\omega$ is a token sequence and $p$ is its probability of being a prefix of a string drawn from the distribution induced by $\llm$. At each iteration, a string $\omega$ is dequeued (Line 5). If the string is $\endtok$-terminated, we remove $\endtok$ and return the string if it satisfies $\stringConstraint$ (Lines 6-7). If the string is not $\endtok$-terminated and is completable with respect to $\stringConstraint$, we enqueue its one-token continuations for later review (Lines 8-10). Unconstrained decoding corresponds to removing the {\color{ForestGreen}green} parts of the algorithm.\label{fig:decoding}
}
\end{figure*}

To abstract over the choice of the decoding strategy,
\Cref{fig:decoding} uses an opaque $\wl$ data structure that stores all generated token sequences whose proper prefixes are completable, along with the sequences' generation probabilities.
In each iteration, the algorithm dequeues a sequence $\omega$ from the worklist. (Note that for sampling-based techniques dequeuing may be a probabilistic operation.)
If $\omega$ is $\endtok$-terminated, we return $\omega[:-1]$ if $\omega[:-1]$ satisfies $\stringConstraint$ and discard $\omega$ otherwise (lines 6-7). If $\omega$ is not $\endtok$-terminated, we check whether any completion of $\omega$ can satisfy $\stringConstraint$ (line 8). If yes, we enqueue the possible continuations $\omega\tau$ for later review (lines 9-10). If no, we discard $\omega$.

% By varying the behavior of $\enqueue$ and $\dequeue$, \Cref{fig:decoding} can encode various exploration strategies,
% such as greedy decoding with or without backtracking,
% beam search, etc.,
% which are orthogonal to the contributions of this paper.
% As the complexity of the constraint increases, this approach becomes computationally expensive and brittle~\cite{typespaper}.
% might store some number of top-probability completions (like in beam search)
% or might recreate the next candidate on the fly from $\omega$ (like in backtracking search).

% \loris{todo add figures side by side of syntactic and semantic CD}
% \Cref{alg:constrained-decoding-synt} outlines the standard constrained decoding loop for the well-studied case in which a constraint is imposed over the \textit{syntactic} structure of the output sequence using a formal language of valid sequences---e.g., via a grammar.
% %
% Line \loris{todo}, is the key point in which one asks whether there exists a sequence \loris{todo} that can lead to a valid completion that belongs to the target language of completions.
% %
% For most formal languages---e.g., regular expressions and grammars---it is not too complicated to retrofit existing parsing algorithms to perform the check at Line \loris{todo} and many such approaches have been proposed~\cite{}.
% %
% \loris{what else do we need here to make rest of section self contained, might want to talk about discrepancy between LLM tokens and Grammar tokens at least?}
% \np{Agreed}

\subsubsection{Soundness and Relative Completeness of Constrained Decoding}
\label{sec:soundness-completeness}
The guarantees we can make about \Cref{alg:constrained-decoding} depend on the quality of $\monitor_\stringConstraint$.
A monitor is \rone \emph{sound} if every prefix the monitor accepts is completable:
        $\monitor_\stringConstraint(\omega) \Rightarrow \exists \omega' \in \alphabet^*. \stringConstraint(\omega \omega')$; and
    \rtwo \emph{complete} if all completable prefixes are accepted by the monitor:
        $\monitor_\stringConstraint(\omega) \Leftarrow \exists \omega' \in \alphabet^*. \stringConstraint(\omega \omega')$.
Constrained decoding is \emph{sound} if every string produced by \Cref{alg:constrained-decoding} satisfies the constraint $\stringConstraint$.
Constrained decoding is \emph{complete} if every string satisfying $\stringConstraint$ can be returned by \Cref{alg:constrained-decoding} with non-zero probability.
\begin{lemma}[Soundness]\label{lem:cd_sound}
    Constrained decoding is sound if $\stringConstraint(\omega)$ can be checked soundly.
\end{lemma}
This is easy to see, as the algorithm only returns outputs $\omega$ that pass the check $\stringConstraint(\omega)$.
Note that this condition ensures that \Cref{alg:constrained-decoding} is sound \textit{even when the monitor is unsound}.

Unlike soundness, completeness often does not hold for common instantiations of $\wl$ (e.g., greedy decoding, where the highest probability continuation of the previously dequeued string is always selected). 
However, we still want a relative characterization of completeness to ensure that constrained decoding cannot discard solutions we could achieve with unconstrained decoding.
To this end, we define a decoding strategy as simply an implementation of $\wl$.
We say constrained decoding is \emph{complete up to a decoding strategy} if under that implementation, \Cref{alg:constrained-decoding} can return every valid solution that unconstrained decoding can produce with at least as high probability.

We prove completeness up to two of the most popular decoding strategies, top-$p$ and top-$k$ sampling (though our proof methodology generalizes to other common strategies).
Top-$k$ sampling is a decoding method where, at each step, the model considers only the 
$k$ most likely next tokens for its current prefix and randomly selects one based on their normalized probabilities.
Top-$p$ sampling is similar, except instead of a fixed $k$ we use the smallest set of most probable tokens whose cumulative probability mass is above threshold $p$.
Either of these sampling approaches could be implemented as instantiations of \Cref{alg:constrained-decoding} by setting dequeue to iterate over the worklist, only choose pairs where the token sequence has a maximal length, and sample from the appropriate subset after renormalizing probabilities.

\begin{lemma}[Completeness up to Top-$p$ and Top-$k$ Sampling]\label{lem:cd_complete}
    If $\monitor_\stringConstraint$ is complete, then constrained decoding (\Cref{alg:constrained-decoding}) is complete up to both top-$p$ and top-$k$ sampling.
\end{lemma}

Note that not all decoding strategies satisfy completeness up to decoding. In a greedy/deterministic beam search for example, discarding an incompletable prefix might cause the LM to explore a completable prefix it would ordinarily have neglected. This exploration may cause constrained decoding to find a different solution than unconstrained decoding would have originally found.

% In our evaluation, following prior work~\cite{typespaper},
% we use top-$p$ sampling for $p=1$, i.e., we repeatedly sample without backtracking.
%
While complete monitors guarantee completeness up to common decoding strategies, sound monitors make search more efficient.
\label{sec:syntactic-constraint-decoding}
In existing approaches for \textit{syntactic constrained decoding}, the string constraint $\stringConstraint$ is limited to asserting membership in a context-free language. In this case, monitors can be sound and complete~\cite{greatgramma,syncode}.

\subsection{Semantic Constrained Decoding}
\label{sec:semantic-constrained-decoding}
\emph{Semantic constrained decoding} is the special case of constrained decoding where the string constraint is defined as $\stringConstraint = \ASTConstraint \circ \parse$ for some predicate $\ASTConstraint\colon \AST \to \bool$. We call $\ASTConstraint$ a \emph{semantic constraint}, We implicitly extend the domain of any such $\ASTConstraint$ to $\AST \cup \{\bot\} \to \bool$ by setting $\ASTConstraint(\bot) = \false$.
For example, the semantic constraint in \Cref{sec:overview} checks whether an AST is an odd sum of even literals. 
The corresponding string constraint checks whether a string parses to such an AST.

\Cref{alg:constrained-decoding} requires a monitor for checking completability of the string constraint $\stringConstraint$.
For semantic constrained decoding, we have that
$$\completable(\omega, \stringConstraint) \equiv \exists \omega'.~ \ASTConstraint(\textsf{parse}(\omega \omega'))$$
%
% As we explained in \Cref{sec:introduction}, defining such a monitor directly is not straightforward.
%
The key insight of our approach is that we can re-express $\completable(\omega, \stringConstraint)$ in terms of the \textit{space} of ASTs a prefix $\omega$ represents.
First, we define the \textit{prefix space} of a string $\omega$ as the set of all parses of continuations of $\omega$:
$$\prefixSpace(\omega) \triangleq \{\textsf{parse}(\omega\omega') \mid \omega' \in \Sigma^*\}$$
By changing the quantified variable to range over ASTs, completability can now be written as
$$\completable(\omega, \stringConstraint) \equiv \exists t \in \prefixSpace(\omega).~ \ASTConstraint(t)$$
We are now asking if there exists a term in a given prefix space that satisfies a constraint $\ASTConstraint$.
Readers familiar with program synthesis will recognize this formulation as the standard shape of a synthesis problem.
Specifically, such a query is an example of a \emph{realizability}~\cite{unrealizability2020,popl23unrealizabilitylogic} problem, which asks whether a set contains a term that satisfies a given property.
$$\realizable(T, \ASTConstraint) \triangleq \exists t \in T.~\ASTConstraint(t)$$
Expressing realizability entirely over spaces of ASTs  removes the need for the string-level reasoning that has characterized prior constrained-decoding work~\cite{typespaper, monitors}. 
All together, our definition of completability simplifies to:
$$\completable(\omega, \stringConstraint) \equiv \realizable(\prefixSpace(\omega), \ASTConstraint)$$ 
Consequently, a monitor for $\stringConstraint$ can be implemented as $\realchecker_\ASTConstraint \circ \prefixSpace$, where 
$\realchecker_\ASTConstraint$ is a potentially approximate procedure for checking realizability. 
%
% In \Cref{sec:corecursion}, we explain how the same formalism enables us to both compute $\prefixSpace$ exactly\footnote{In fact, $\prefixSpace(\omega)$ is a regular tree language.} and implement realizability checkers.

\subsubsection{Soundness and Relative Completeness of Semantic Constrained Decoding}
\label{sec:semantic-soundness-completeness}
A realizability checker is \rone \emph{sound} if every space it accepts is realizable:
$\realchecker_\ASTConstraint(T) \Rightarrow \exists t \in T.~\ASTConstraint(t)$ and
\rtwo \emph{complete} if every realizable space is accepted by the checker: ${\realchecker_\ASTConstraint(T) \Leftarrow \exists t \in T.~ \ASTConstraint(t)}$.

Since $\prefixSpace$ can be computed exactly, the following theorems follow immediately from \Cref{lem:cd_sound} and \Cref{lem:cd_complete}:
\begin{theorem}[Soundness of Semantic Constrained Decoding]\label{thm:sem_cd_sound}
     Semantic constrained decoding is sound if $\ASTConstraint$ can be checked soundly.
\end{theorem}

\begin{theorem}[Completeness up to Top-$p$ and Top-$k$ Sampling of Semantic Constrained Decoding]\label{thm:sem_cd_complete}
    If $\realchecker_\ASTConstraint$ is complete, then semantic constrained decoding is complete up to both top-$p$ and top-$k$ sampling.
\end{theorem}

\Cref{sec:overview} showed an example of a sound and complete realizability checker for the constraint being enforced.
In \Cref{sec:case-studies}, we present two more realistic examples:
a sound and complete realizability checker for the problem of equivalence-guided decoding,
and a complete\footnote{For this problem one can prove that no sound and complete realizability checker exists~\cite{system_f_untypable}} realizability checker for type-safe decoding.

\section{\name: Semantic Constrained Decoding via Coinduction}
\label{sec:corecursion}
The key idea of semantic constrained decoding is that one can define pruners on \textit{spaces} of programs, which are defined at the level of abstract syntax.
To create a programmatic interface for defining pruners, we need a framework for constructing, manipulating, and performing analyses over possibly infinite program spaces.
One framework for cleanly manipulating infinite objects is supplied by \citet{cocaml}; the framework  defines infinite objects with repeated structure (like program spaces) as \emph{regular codata} and lets users compute recursive functions on regular codata (\Cref{sec:coinduction-intro}).
We adopt this framework to represent and manipulate program spaces (\Cref{sec:program-spaces-as-coterms}).
As we show in \Cref{sec:parsers}, framing program spaces as regular codata allows us to not only implement pruners over program spaces, but also gives an elegant way to construct prefix spaces.
%
% In \Cref{sec:coinduction-intro}, we give a brief tour of how to program in this framework.
%
Finally, we describe how \name implements this approach as an embedded DSL in Python (\Cref{sec:implementation}).
%

% We refer the interested reader to \cite{cocaml}.

\subsection{Programming with Regular Coinduction} \label{sec:coinduction-intro}
Coinduction is a powerful technique for representing and manipulating infinite objects. 
We present a brief introduction to coinduction adapted from the running example used by \citet{cocaml}.

\subsubsection*{Codata}
Consider the usual datatype for integer lists:
\begin{lstlisting}
data IList = INil
           | ICons Int IList
\end{lstlisting}
The \code{IList} datatype includes finite lists like \code{INil} and \code{ICons 1 (ICons 2 INil)}. However, \code{IList} also includes \emph{infinite} streams like, e.g., \code{ICons 1 (ICons 2 (ICons 3 ...))}. Infinite members of a datatype are called \emph{coterms}, or sometimes \emph{codata}. 

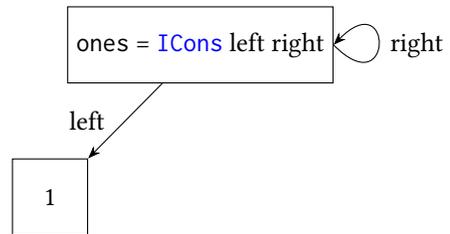
\begin{wrapfigure}{r}{0.45\textwidth}
  \vspace{-10pt}
  \centering
  \begin{tikzpicture}[node distance=2cm, every node/.style={draw, minimum size=1cm}, ->, >=Stealth]

    % Nodes
    \node (one) at (0,0) {1};
    \node (ones) at (2,2) {\code{ones} = \code{ICons} left right};

    % Edges
    \draw[->] (ones) -- (one) node[midway, left, draw=none] {left};

    % Self loop
    \draw[->] (ones.east) .. controls +(0.8,-0.8) and +(0.8,0.8) ..
      (ones.east) node[midway, right, draw=none] {right};

  \end{tikzpicture}
  \vspace{-10pt}
  \caption{Structure of the \code{ones} stream.}
  \label{fig:ones}
\end{wrapfigure}
In general, it is not always possible to finitely represent a coterm. However, many coterms do admit compact representations. For example, the stream of all ones can be written as \code{ones = ICons 1 ones}. The stream \code{ones} is an infinite term, but it only has finitely many distinct subterms (i.e., \code{1} and \code{ones}). This property, called \emph{regularity}, allows codata to be expressed with a finite representation. Alternatively, regular coterms can be depicted as cyclic term graphs, as shown by \Cref{fig:ones}.

Because regular codata can be finitely represented, we can compute their properties in finite time. For example, one cannot compute the set of integers contained in an arbitrary stream. However, a quick inspection of the definition of \code{ones} (or a fixpoint procedure over the term graph) reveals that \lstinline{ones} contains no integers besides $1$. When we represent program spaces as regular codata later, we will use these nice properties to answer questions about, e.g., the nonemptiness of a program space.
% \loris{I assume there is some known relationship to tree automata we can talk about?}

\subsubsection*{Computations over Regular Codata} 
% That's objects; what about functions?
% You can define functions on codata as well, but you need to be careful that you can actually evaluate them.
% Cocaml lets us evaluate 2 kinds of functions.
% Functions that map regular codata to regular codata
% Example
% Functions that compute fixpoints over regular codata (with a bounded-height lattice) 
% Example
% This is a powerful formalism; now let's see how we use it to do cool stuff.

To describe operations on codata, we can define typical recursive functions on the structure of the datatype representing codata. For example, we can define the following function to collect the set of distinct integers appearing in a stream:
\begin{lstlisting}
@fixpoint: EmptySet
elements :: IList -> Set Int
elements INil = EmptySet
elements ICons x xs = (add (elements xs) x)
\end{lstlisting}
Over an arbitrary infinite coterm \code{t}, it may not be possible to explicitly compute \code{(elements t)}. However, one can do so when \code{t} is regular. The key observation is that, if \code{t} has only finitely many distinct subterms, we will only need to evaluate \code{elements} on finitely many distinct coterms. In general, functions that compute data from regular coterms can be evaluated as fixpoint computations over the coterm's cyclic term graph. Throughout the paper, we distinguish such functions with the annotation \code{@fixpoint}, followed by the initial value for the fixpoint computation.

It is also possible to
% It is possible to recursively 
recursively define functions that produce
% map regular coterms to 
regular coterms.
In the case of streams, we might consider a function that adds $1$ to each element of the stream
\begin{lstlisting}
addOnes :: IList -> IList
addOnes INil = INil
addOnes (ICons x l) = ICons (x + 1) (addOnes l)
\end{lstlisting}

If we apply \code{addOnes} to \code{ones}, we see that only finitely many distinct terms are generated:
\begin{lstlisting}
addOnes ones = addOnes (ICons 1 ones)
             = ICons 2 (addOnes ones)
\end{lstlisting}

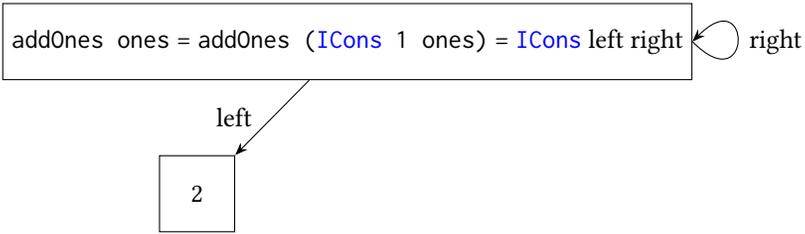
\begin{figure}
% {r}{0.7\textwidth}
  \centering
  \begin{tikzpicture}[node distance=2cm, every node/.style={draw, minimum size=1cm}, ->, >=Stealth]

    % Nodes
    \node (one) at (0,0) {2};
    \node (ones) at (2,2) {\code{addOnes ones} = \code{addOnes (ICons 1 ones)} = \code{ICons} left right};

    % Edges
    \draw[->] (ones) -- (one) node[midway, left, draw=none] {left};

    % Self loop
    \draw[->] (ones.east) .. controls +(0.8,-0.8) and +(0.8,0.8) ..
      (ones.east) node[midway, right, draw=none] {right};

  \end{tikzpicture}
  \vspace{-10pt}
  \caption{Structure of \code{addOnes ones}.}
  \vspace{-7pt}
  \label{fig:addones}
\end{figure}

Regular corecursive functions (i.e., functions that map regular coterms to regular coterms) can be computed explicitly over regular coterms with the techniques of \citet{cocaml}.
At a high level, their techniques work by memoizing the result of function expansions and producing a cycle upon encountering a previously seen call.

Of course, it is possible that the user may write a function that yields a coterm that is \textit{not} regular.
We do nothing to prevent this or otherwise enforce termination---nontermination is a wholly orthogonal problem to our purposes that has been studied in the literature on coinduction, for example in~\cite{capretta2005general}. In practice, one can avoid most issues by \rone only writing recursive functions that call themselves on subterms of the input term and
\rtwo avoiding \code{@fixpoint}-annotated functions whose return types contain infinitely many values.
Finally, our framework provides one special built-in function, \code{isCyclic}, that returns true if a regular coterm has cycles and false if it is just a tree.

\subsection{Program Spaces as Coterms}
\label{sec:program-spaces-as-coterms}
In \Cref{sec:overview} we lifted the \code{Expr} datatype, which describes individual programs, to \code{ExprSpace}, which describes spaces of programs. Now, we generalize that principle to arbitrary program datatypes.
First, we define a generic \code{Program} datatype (\Cref{fig:program-type}) that describes the structure of abstract syntax trees.
It is parametric over a type parameter \code{sym}, which represents the different possible constructors for ASTs.
We write \code{Program} this way to allow for manipulations over ASTs that are generic over the constructor.
Often, we want to treat a datatype like \code{Expr} ``as if it were'' an instance of \code{Program}. 
That is, \code{Expr} is roughly the same as \code{Program ExprSym}, where \code{ExprSym = Sum}, ignoring constraints on arity of children.
For notational convenience, we ignore this distinction and treat these two datatypes as the same.

\changed{
We can abstract over \code{ExprSpace} in the same way we abstracted over \code{Expr}. 
Specifically, we can define a \code{ProgramSpace} (\Cref{fig:programspace-type}) as either:
\code{Empty}, the empty space representing no programs; \code{Union}, a union of program spaces; \code{Lit}, a set of concrete literals described by a regular expression; or \code{ASTNode}, a product of program spaces under some AST constructor.
}
We will write most concrete program spaces in the style of \code{ExprSpace} for clarity, but generic algorithms on program spaces will be written over \code{ProgramSpace}.

As discussed in \Cref{sec:overview}, we need to be able to decide whether a \code{ProgramSpace} represents the empty set. While \code{Empty} is obviously empty, determining emptiness of a coterm like\linebreak 
\code{E = Union [E, E]} requires a fixpoint computation. To solve this problem, we implement a simple function, \code{nonempty}, to check whether a space is empty:
\begin{lstlisting}
@fixpoint: False
nonempty :: ProgramSpace sym -> ProgramSpace sym
nonempty Empty = False
nonempty Union children = any nonempty children
nonempty (Lit l) = regexNonempty l
nonempty ASTNode _ children = all nonempty children
\end{lstlisting}
% \begin{lstlisting}
% @fixpoint: False
% nonempty :: ProgramSpace t sym -> ProgramSpace t sym
% nonempty Empty = False
% nonempty Union children = any nonempty children
% nonempty (Lit l) = regexNonempty (snd l.td)
% nonempty App _ children = all nonempty children
% \end{lstlisting}

Finally, when writing pruners it will be useful to convert a \code{ProgramSpace} representing only a single program into the corresponding element of type \code{Program}.
Such spaces naturally arise when prefixes contain complete subexpressions.
For this reason, we define a utility function \code{collapse} as:
\begin{lstlisting}
collapse :: ProgramSpace sym -> Maybe (Program sym)
collapse x | isCyclic x     = Nothing
collapse (Union [x])        = collapse x
collapse (Lit l)
  | isSingleton l           = Just (ASTLeaf (toString l))
collapse ASTNode sym [children] = do
    collapsedChildren <- mapM collapse children
    ASTNode sym collapsedChildren
collapse _                  = Nothing
\end{lstlisting}

\subsection{Parsers as Coterms} \label{sec:parsers}
Recall from \Cref{sec:overview} that our lexer takes a string $\omega$ and returns a finite set of \emph{lexical prefixes}. A lexical prefix is a sequence of complete lexemes followed by a possibly incomplete lexeme. 
The lexer is designed so that every lex of a completion of $\omega$ extends one of the lexical prefixes.
% The sequences of lexemes that extend some lexical prefix in the lexer's output are exactly the lexes of continuations of $\omega$. 
The details of our lexing algorithm are discussed in \Cref{app:lexing}.
% \timothychanged{\sout{Note that we define \code{Lexeme = String}. The ``type'' of a lexeme is implied by its string value because the regexes associated to different kinds of lexemes (e.g., int, string, identifier) are disjoint.}}

Once lexing is complete, we must perform parsing. We carry out parsing by following the \textit{parsing with derivatives} technique introduced by \citet{parsingWithDerivatives}. Intuitively, a \emph{parser} is a function from strings to trees that represents the state of a parser after consuming some string $\omega$. Parsers support a \emph{derivative} operation that advances the parser state by consuming a lexeme, analogous to Brzozowski derivatives for regular expressions~\cite{Brzozowski}. Parsers, their derivative operation, and a method to convert a parser to a program space can be easily represented as codata and corecursive functions.
%
% \timothychanged{\sout{Because parsers are rather difficult to understand on first pass and not the focus of this paper, we give a brief introduction and refer readers to \cite{parsingWithDerivatives} for a deeper discussion.}}

%
Intuitively, after consuming some prefix $\omega$, a parser (\Cref{fig:parser-states}) can be in one of a few possible states when it considers how to handle further lexemes. 
\begin{itemize}
    \item \code{EmptyParser}: The parser accepts no strings, meaning $\omega$ is not in the language that the parser recognizes.
    \item \code{Choice ps}: The parser accepts any string accepted by any parser in the list of parsers ps.
    \item \code{ConstParser reg}: The parser accepts any single lexeme in the regex reg.
    \item \code{NullParser l}: The parser has parsed the lexeme \code{l} but will 
    consume nothing further.
    % do nothing else.
    \item \code{Seq a ps1 ps2}: The parser is a sequential composition of parsers, where \code{ps1} is the sequence elements whose parsing is complete and \code{ps2} is the sequence of elements that remain. The \emph{action} \code{a} describes how completed parses map to AST nodes.
\end{itemize}
\begin{figure}
  \centering
  \begin{subfigure}[t]{0.48\textwidth}
    \centering
\begin{lstlisting}[xleftmargin=0em]
data Program sym =
  | ASTLeaf String
  | ASTNode sym [Program sym]
\end{lstlisting}
    \caption{\code{Program} datatype, parametric over a type \code{sym} of AST constructors.}
    \label{fig:program-type}
  \end{subfigure}
  \hfill
  \begin{subfigure}[t]{0.48\textwidth}
    \centering
\begin{lstlisting}[xleftmargin=0em]
data ProgramSpace sym =
  | Empty
  | Union [ProgramSpace sym]
  | Lit Regex
  | ASTNode sym [ProgramSpace sym]
\end{lstlisting}
    \caption{Lifting of the \code{Program} datatype to a \code{ProgramSpace}.}
    \label{fig:programspace-type}
  \end{subfigure}
  \caption{Programs and program spaces.}
  \label{fig:program-datatypes}
\end{figure}
\subsubsection{Derivatives on Parsers}
To advance the state of a parser, we define a \emph{corecursive derivative function} \code{d} (\Cref{fig:parser-derivative}). Given a lexeme $l$ and a parser $p$, the coterm \code{d l p} represents the parser that is obtained after the parser \code{p} consumes the lexeme $l$.
\changed{
The definition of \code{d} recurses on the structure of the input parser \code{p}:
\begin{itemize}
    \item If $p$ is an \code{EmptyParser}, then \code{d p l} returns \code{EmptyParser}.
    \item If $p$ is a \code{NullParser}, then \code{d p l} returns \code{EmptyParser} as $p$ cannot accept further input.
    \item If $p$ is a \code{ConstParser}, then \code{d p l} returns a constant parser storing $l$ if the regex for $p$ matches $l$; otherwise it returns \code{EmptyParser}.
    \item If $p$ is a \code{Choice}, then \code{d p l} returns a \code{Choice} of the derivative of its children.
    \item If $p$ is a \code{Seq}, then \code{d p l} derives the parser at the current position in the sequence and splits into a \code{Choice} of cases. If the derived parser can accept the empty string, we advance to the next parser in sequence by applying the nullability operator \code{delta}, which “completes” a parser by throwing away all parses that haven’t been fully matched. Otherwise, we update the sequence in place.
\end{itemize}
}
\begin{figure}
    \centering
\begin{subfigure}[t]{0.48\textwidth}
\begin{lstlisting}[xleftmargin=0em]
data Parser sym =
  | EmptyParser
  | Choice [Parser sym]
  | ConstParser (Regex reg)
  | NullParser (String l)
  | Seq (Action sym) 
        [Parser sym] [Parser sym]
\end{lstlisting}
\caption{\code{Parser} datatype, parametric over the the type parameter \code{sym} describing the AST constructors it can produce.}
\end{subfigure}
\hfill
\begin{subfigure}[t]{0.48\textwidth}
\begin{lstlisting}[xleftmargin=0em]
data Action sym = Action {
    constructor : Maybe sym,
    child_indices: [Int]
}
\end{lstlisting}
\caption{An \code{Action} describes how a sequence of parse trees is converted into an abstract syntax tree: it stores the AST constructor to use along with a list of indices of children to pass in.}
\end{subfigure}
\caption{Parsers and actions.}
\label{fig:parser-states}
\end{figure}
\begin{figure}
  \centering
  \begin{subfigure}[t]{0.48\textwidth}
    \centering
\begin{lstlisting}[xleftmargin=0em]
d :: Lexeme -> Parser sym -> Parser sym
d l (ConstParser reg)      =
  if matches reg l
    then NullParser l     
    else EmptyParser
d l (Choice ps)            = 
    Choice (map (d l) ps)
d l (Seq a finished p::ps) =
  let derived = d l p in
  Choice [
    Seq a (finished ++ [delta derived]) ps,
    Seq a finished (derived : ps)
  ]
d _ _                      = EmptyParser
\end{lstlisting}
    \caption{Derivative for parsers: in the \code{Seq} case, we use the helper $\delta$ to handle the case where the current parser in sequence has finished consuming input.}
    \label{fig:parser-derivative}
  \end{subfigure}
  \hfill
  \begin{subfigure}[t]{0.48\textwidth}
    \centering
\begin{lstlisting}[xleftmargin=0em]
delta :: Parser sym -> Parser sym
delta p@(NullParser _)  = p
delta (Choice ps)       = 
    Choice (map delta ps)
delta p@(Seq _ _ [])    = p
delta _                 = EmptyParser
\end{lstlisting}
    \caption{The function $\delta$ returns the ``nullable'' portion of a parser, i.e. the branches which do not consume further input.}
    \label{fig:parser-delta}
  \end{subfigure}
  \caption{Definition of derivative for parsers.}
  \label{fig:parser-defs}
\end{figure}
More details appear in \citet{parsingWithDerivatives} and \citet{parsingWithZippers}.

\subsubsection{From Parsers to Program Spaces}
Finally, we need to construct a program space reprsenting all programs our parser could emit. To do so, we define a corecursive function \code{image} (\Cref{fig:image_def}) that recurses over the structure of the input parser \code{p}:
\changed{
\begin{itemize}
    \item If $p$ is an \code{EmptyParser}, then \code{image} returns the empty space.
    \item If $p$ is a \code{Choice}, then \code{image} returns a \code{Union} of the images of its children.
    \item If $p$ is a \code{ConstParser}, then \code{image} returns a \code{Lit} matching any lexeme $p$ could accept.
    \item If $p$ is a \code{NullParser}, then \code{image} returns a \code{Lit} matching only the lexeme $p$ has already accepted.
    \item If $p$ is a \code{Seq}, then \code{image} applies the \code{Action} to $p$. If the constructor for the action is \code{None}, we ``skip'' converting this layer of the parser---this option supports inlining trees for constructs like parenthesis specifically.
\end{itemize}
}

\begin{figure}
    \centering
\begin{lstlisting}
image :: Parser sym -> ProgramSpace sym
image EmptyParser = Empty
image (Choice ps) = Union (map image ps)
image (ConstParser reg) = Lit reg
image (NullParser l) = Lit l
image (Seq a ps ps') =
    let children = ps @ ps' in
    match a.constructor with
        Nothing ->
            image $ children!!(head a.child_indices)
        Some s  ->
            ASTNode s [image $ children!!i | i <- a.child_indices]
\end{lstlisting}
    \caption{The image of a parser is the space of trees it could emit.}
    \label{fig:image_def}
\end{figure}

Together, the lexing, derivative-based parsing, and image construction enable \toolname to track the evolving set of valid program completions at every step of generation. This coinductive representation forms the foundation for subsequent semantic pruning and realizability checking.

\subsection{Implementation}
\label{sec:implementation}
We implemented \name as an embedded domain-specific language (DSL) in Python that supports symbolic manipulation of coinductive structures. The backend draws inspiration from prior work on regular coinduction, such as CoCaml~\cite{cocaml}, and makes extensive use of Python metaprogramming.
We do not use CoCaml directly to ensure compatibility with the rest of our Python implementation.
Notably, \name defines two key decorators:
\begin{itemize}
    \item
\code{@rewrite}: Marks functions that construct coterms by symbolically rewriting recursive calls into a finite system of equations (if the recursion is regular).
\item
\code{@fixpoint}: Interprets a coterm via Kleene iteration to compute a fixed point over the coterm's input.
\end{itemize}

Additionally, \name incorporates domain-specific simplifications for program spaces.
For instance, unions over parser states or AST spaces are simplified by removing empty branches---i.e., when  \code{nonempty} returns false.
For parser derivatives, we also implement the compaction optimization described in ~\cite{parsingWithDerivatives} to reduce the size of coterms.
\section{Case Studies}\label{sec:case-studies}

We demonstrate the generality of our framework by instantiating it for two domains:
enforcing semantic equivalence for a basic functional language and enforcing type safety for TypeScript.

\subsection{Equivalence-Guided Decoding}\label{sec:case-studies:egraphs}
In this case study, the LM is given expressions in a basic functional language and is asked to refactor them into equivalent programs.
The language consists of basic arithmetic operators, identifiers, integer constants, function applications, and let bindings.
For example, consider the following program that computes the distance between two points.
\begin{lstlisting}
    sqrt (pow (x2 - x1) 2 + pow (y2 - y1) 2)
\end{lstlisting}
A valid output for the LM might be:
\begin{lstlisting}
    let dx = x2 - x1 in
    let dy = y2 - y1 in
    sqrt (pow dy 2 + pow dx 2)
\end{lstlisting}
\changed{
In general, program equivalence is not decidable; therefore, \toolname uses \textit{e-graphs}~\cite{egg} to reason about program equivalence approximately.
An e-graph is a data structure for compactly representing spaces of equivalent programs.
At a high level, an e-graph maintains a collection of equivalence classes of expressions.
To build an e-graph, a user supplies an initial expression and a list of \textit{rewrite rules} that define equivalences.
For example, the rule $?x + ?y \to ?y + ?x$ encodes commutativity of addition.
A term like $f(a, b) + 2$ can \textit{match} the left-hand side of this rule, to be rewritten to $2 + f(a, b)$.
Initially, an e-graph only represents trivial equivalences (i.e., every equivalence class has only one member).
Equivalence classes are then refined via a process known as \textit{equality saturation}, which iteratively applies rules (up to a fixed budget) to subterms in the e-graph.
Whenever a subterm matches a rule, the e-graph is modified by applying the rewrite and merging the equivalence classes of the rewritten subterm with the subclass of the original subterm.
Crucially, e-graphs can compactly represent large sets of terms by \textit{sharing}: instead of storing a term like $f(x, y)$ directly, an e-graph will add a layer of indirection so that the node $f$ has pointers to the \textit{equivalence classes} of $x$ and $y$, rather than the expressions themselves.
}

For our case study, we define a set of basic rewrite rules for arithmetic expressions (ones with addition, subtraction, multiplication, and division; see \Cref{app:egraphs_results} for the full set of rules).
We then use the egglog library~\cite{egglog} to build an e-graph. 
\toolname then enforces the constraint that the output of an LM must be a program stored by this e-graph---\changed{i.e., \toolname guarantees that the output term is equivalent up to rewrite}.

Importantly, the set of terms an e-graph represents forms a regular tree language that can be represented by a finite tree automaton \cite{suciu2025semantic}.
We use this property to build a semantic pruner whose goal is to only keep terms that appear in the automaton.
In particular, we syntactically convert the automaton corresponding to the e-graph into a corresponding \code{ProgramSpace} coterm.
Then we can define a function \code{intersectSpaces} to filter out non-equivalent programs from the prefix space.

While e-graphs can easily deal with let-free expressions, we need some extra care to make sure that let-expressions are taken into account by the e-graph.
To handle let expressions, our pruner adopts a simple strategy: if we have not finished generating a let binding, we do not prune any programs---no matter what expression the let binds a variable to, we can always avoid using the variable in the final expression.
If a let expression space contains a complete binding of form $x = e_1$ where $e_1$ is a complete expression (i.e., the LLM has generated a let binding), then we update the current e-graph with an additional rewrite rule $e_1 \to x$.
When generating the last expression under the let expression, we use the semantic pruner induced by this new e-graph.
The pseudocode for our checker is shown below:
\begin{lstlisting}
prune (Let b s1 s2) egraph = fromMaybe Empty do
    x <- collapse b       -- Wait for variable name to be generated
    e1 <- collapse s1     -- Wait for binding for x to be generated
    let egraph' = updateEGraph egraph x e1
    prune s2 egraph'      -- Recurse, adding extra rule e1 => x
prune s egraph             = s `intersectSpaces` (eGraphToSpace egraph) 
\end{lstlisting}

\subsection{Type-Safe Decoding for TypeScript}\label{sec:case-studies:types}
In our second case study, we use \toolname to implement a semantic pruner that enforces type safety for a subset of TypeScript that covers most core language features (branching, loops, variable declaration and reassignment, constant and mutable variables, and [recursive] function definitions). We study TypeScript because it is widespread~\cite{octoverse2022}, and previous work has shown that small and mid-sized LMs often fail to produce type-safe TypeScript code~\cite{typespaper}.

It is a folklore result that even for simply-typed lambda calculus, the set of well-typed terms is not a regular tree language. As a consequence, no pruner can remove exactly the ill-typed programs from a program space---the pruned space would not be representable as regular codata. Instead, we implement an overapproximate pruner that conservatively discards many ill-typed programs while ensuring that all well-typed programs are preserved.

% For this case study, we restrict our attention to a syntactic subset of TypeScript that \snchanged{includes core language features  but} omits certain features such as strings, arrays, lambda abstractions, and property accesses. \toolname imposes no barriers to supporting the omitted constructs\timothychanged{; however a full treatment of TypeScript is beyond the scope} of this paper.
% Our TypeScript fragment still captures many core language features, including branching, loops, variable declaration and reassignment, constant and mutable variables, and (recursive) function definitions.

\paragraph{Type Pruning}
Our typesafety pruner is a corecursive function of the following form:
\begin{lstlisting}
type_prune :: Environment -> Type -> ProgramSpace -> ProgramSpace
\end{lstlisting}
Much like a type checker, \code{type_prune} is parameterized by a typing environment and target type. However, it operates on program spaces instead of single programs: given an input space it returns an overapproximation of the subspace of well-typed programs.
For example, for the program space \code{(Union (Add x 6) false)}, the expression
\begin{lstlisting}
type_prune {x: int} int (Union (Add x 6) false) = Add x 6
\end{lstlisting}
yields a program space containing \code{(Add x 6)}, the only program that types to \code{int} when the variable \code{x} has type \code{int}.
% That is, if we prune the space \code{Union (Add x 6) false} to only include programs of target type \code{int} (and \code{x} has type \code{int}), we get \code{Add x 6}, a program space containing the only program that types to \code{int} (when \code{x} has type \code{int}).
If we call \typerestrict with the empty environment and the target type $\top$ (i.e., the supertype of all types), we get an approximation of the set of well-typed closed programs---for the program space \code{(Union (Add x 6) false)}, we would get \code{false}. 
% That is exactly what we need to enforce typesafety. 
% The challenge now is to produce a syntax-directed implementation of \typerestrict that can effectively prune programs in a program space.
\Cref{app:typescript} discusses the full implementation of \code{type_prune}; we give an overview below.

\paragraph{Designing \typerestrict}
We designed \typerestrict to be an analogue of checking in bidirectional type checking~\cite{bidirectional}, which resolves typing judgements locally and is deterministic.
For ground terms, we check their type against the target type and decide whether to keep them or prune them (yielding the \code{Empty} program space):
\begin{lstlisting}
    type_prune {x: int} int  6 = 6
    type_prune {x: int} bool 6 = Empty
\end{lstlisting}
For most program constructors, a bidirectional type checker would check the types of the subterms to decide how to type the superterm. Similarly, our type pruner prunes each set of subterms and returns only the valid combinations: 
\begin{lstlisting}
type_prune {x: int} int (Add x 6) = 
    Add (type_prune {x: int} int x) (type_prune {x: int} int 6)    
\end{lstlisting}

Like with bidirectional type checking, this approach gets complicated for function applications. Suppose \code{F} and \code{X} are program spaces, and consider type pruning the application of \code{F} to \code{X}---i.e., \linebreak \code{(type_prune   \{\}   int   (App F X))}. We would like to handle this case by pruning \code{F} and \code{X} to their expected types and combining the results. However, there are infinitely many valid combinations of types that programs in \code{F} and \code{X} could take, e.g., \code{int -> int} paired with \code{int}, or \code{(bool -> int) -> int} paired with \code{bool -> int}, and so on.

\changed{In bidirectional type checking, this multiplicity is resolved by first inferring the type of the function. When the type of the function is known, the argument can be checked against the function's input type. Since \code{F} may represent infinitely many programs with infinitely many possible types, type inference on \code{F} may be insufficient to resolve the multiplicity. Instead, we \textit{wait} until \code{F} resolves to a single AST (i.e., via \code{collapse} from \Cref{sec:program-spaces-as-coterms}). Before \code{F} resolves to a single AST, we type prune \code{F} to exclude functions that cannot return the target type, and we leave \code{X} unmodified. Once \code{F} resolves to a single AST, we can infer the type of that AST and use it as the basis of our pruning. For example, if \code{F} = \code{foo}, then}
\begin{lstlisting}
type_prune {foo: (int -> int) -> int} int (App foo X) = 
    App foo (type_prune {foo: (int -> int) -> int} int -> int X)
\end{lstlisting}
\changed{Type inference for single ASTs can be implemented in the usual way (e.g., \cite{typescript_good_typesystem}). type pruning \code{Empty} always yields \code{Empty}, and type pruning a \code{Union} will type prune each child and union the results.}

% btw our approach is also quite tight
% and we bound types
Despite the overapproximate handling of function application, our pruner is quite precise over program spaces that correspond to left-to-right program prefixes. For such spaces, \code{F} will always resolve to a concrete function before any of \code{X} is written. This left-to-right resolution means that we can always infer the type of \code{F} when we start writing \code{X}. 

\changed{We note that \toolname makes it more natural to implement an overapproximate type pruner than an underapproximate type pruner. The idea of waiting for a term to be completed is supported naturally by the \code{collapse} function. In contrast, writing an underapproximate pruner requires analyzing some subset of the combinations of function and argument types at the same time.}

\section{Evaluation}\label{sec:eval}

We evaluate the two instantiations of \name presented in \Cref{sec:case-studies} to answer the following research questions:

\begin{enumerate}[\bfseries RQ1)]
\item \textbf{Effectiveness.} Does \name make language models more effective at generating programs that satisfy semantic constraints?
\item \textbf{Overhead.} What is the computational cost (e.g., latency, token usage) of enforcing semantic constraints during generation?
\end{enumerate}

\subsection{Experiment Setup}

We compare our semantic constrained decoders written in \name to the following baselines:
\begin{itemize}
  \item \emphbf{Unconstrained Decoding}: The LM generates code without any constraints.
  \item \emphbf{Grammar-Constrained Decoding (GCD)}:
   The LM must produce syntactically valid programs (enforced via a grammar), but no semantic restrictions are applied. In \name, grammar-constrained decoding corresponds to using the identity function as the pruner.
\end{itemize}

\subsubsection*{Benchmarks: Equivalence-Guided Decoding}
We created 10 benchmark tasks in the basic functional language shown in \Cref{sec:case-studies:egraphs},
where the goal is to refactor a program into an equivalent one---e.g., factoring out subexpressions into \code{let} bindings.
We use a fixed system prompt that specifies the grammar of the language and instructs the model to return only a refactored program\footnote{We also ran a variant of this benchmark where we included in the prompt the list of rewrite rules. The overall trends are similar; this data can be found in \Cref{app:egraphs_results}.} with no explanation.
The exact prompt and the full list of benchmarks are shown in \Cref{app:egraphs_results}.
We count a response as \emph{correct} if the LM produces a complete program that is equivalent to the input program, with no post-processing.
Unconstrained decoding often fails to produce just the output code, despite being explicitly instructed to do so.
Therefore, we also evaluate a variant where the prompt wraps outputs in triple backticks (\code{```}) to encourage the model to delimit its code clearly. 
This variant avoids penalizing runs that fail only due to formatting.
We refer to the two variants as \textit{No Delimit} and \textit{Delimit}, respectively.

\subsubsection*{Benchmarks: Type-Safe Decoding}
For our second instantiation of \name, we source the benchmarks from the TypeScript translations of the MBPP~\cite{mbpp} tasks
from the MultiPL-E dataset~\cite{multiplE}.
We extract the \changed{72}/809 tasks that can be solved in our language fragment.
\changed{Our prompting strategy mirrors that of \citet{typespaper}, but we additionally}
provide context to the LM to instruct it to avoid language constructs outside our language fragment. Our \changed{full prompting strategy} and benchmarks are included in \Cref{app:typescript_results}.

An example task is given below:
\begin{verbatim}
// Write a typescript function to find the next perfect square
// greater than a given number.
function next_Perfect_Square(N: number): number
\end{verbatim}
We count a response as \emph{correct} if the generated TypeScript program compiles. 
\changed{For the TypeScript evaluation, we also run the tool by \citet{typespaper} for comparison (although we do not instruct the LM to avoid programs outside our language fragment since their tool supports a larger subset of TypeScript).}

\subsubsection*{Models, Parameters and Hardware}

We run all experiments using the instruction-tuned versions (i.e. models that are trained to follow instructions in the prompt) of the following models:
DeepSeek-Coder-6.7b, CodeLlama-7B, and CodeLlama-13B.
Each model is evaluated at five different sampling temperatures: ${0.01, 0.3, 0.5, 0.7, 1.0}$.
These ranges of small-to-medium models and low-to-high temperatures lets us explore a range of model capability.
We run all experiments on a Supermicro SYS-4029GP-TRT with two Intel(R) Xeon(R) Gold 6230 CPUs, 384 GB RAM, a 4 TB SSD, and eight Nvidia Geforce RTX 2080Ti GPUs.

\subsubsection*{Procedure}

For each benchmark, model, decoder, and temperature, we run constrained decoding until either:
the $\endtok$ token is generated or
% \snchanged{\sout{\rtwo
% a fixed token budget (set to 400) is exhausted.
% or 
% \rthree}} 
a 150 second timeout is reached.
We only implement a naive sampling strategy where if a token is rejected, we backtrack by one token and re-sample with that token removed.
Constrained decoders may fail if they use up their token or time budget without completing a valid program---especially if the LM repeatedly proposes unrealizable tokens that must be pruned. 
Unconstrained decoders may also fail to terminate if the model does not emit an $\endtok$ token withing the budget.

A run is considered \textbf{successful} if:
\rone A complete program is emitted, and
\rtwo It satisfies the semantic constraints of the task (i.e., equivalence or type safety).

\subsection{RQ1: Effectiveness}
\subsubsection{Quantitative Evaluation}
\label{sec:quantitative}

\Cref{tab:decoding-performance} reports the number of successful runs for \name and the two baselines---unconstrained decoding and grammar-constrained decoding---on all benchmarks.
Across nearly all configurations, semantic constrained decoding delivers consistent and often dramatic improvements.

\paragraph{Equivalence-Guided Decoding}
Unconstrained, most models perform very poorly on our benchmarks, with most failing to generate even a single semantically equivalent program.
Several factors contribute to this:
\rone We intentionally use small- and medium-sized models, which highlight the gains possible even for less capable models.
% in the absence of scale.
\rtwo The toy language used in the benchmarks is likely out-of-distribution for most pretrained LMs, which are tuned on real-world languages like Python and JavaScript.
In particular, DeepSeek-Coder-6.7B frequently attempts to write Python, Lisp, or TypeScript code, and fails all benchmarks without semantic constraints as a result.
\rthree Despite clear instructions, models frequently emit natural language explanations, markdown, or commentary—none of which are semantically valid outputs under our equivalence checker.

\changed{
Because we define equivalence with respect to a reference set of rewrite rules, it is possible that these rules do not encode every possible equivalence a user might have in mind---e.g., in our benchmarks we treated all named functions as uninterpreted, but particular functions could have additional semantics associated with them.
%
% While we could have simply added more rules covering a wide variety of such cases, one known limitation of e-graphs is that their size can blow up unpredictably with more rules.
%
To mitigate this source of imprecision, we manually check the unconstrained and syntax-constrained data for false negatives (programs that could be considered equivalent, but are not proved equivalent by our rewrite rules). 
We only found 8 such instances out of the 600 total programs for these baselines.
}

\begin{table}[t]
\centering
\caption{Successful generations for different decoding strategies across models and temperatures (higher is better).
For equivalence-guided decoding we report the number of benchmarks for which an equivalent program was produced.
For TypeScript we report the number of benchmarks on which compilable code was produced.
Best results per column are \textbf{bolded}.}
\label{tab:decoding-performance}
\resizebox{\textwidth}{!}{%
\begin{tabular}{@{}ll|ccccc:c|ccccc:c|ccccc:c@{}}
\toprule
 & & \multicolumn{6}{c|}{\textbf{DeepSeek-Coder-6.7b}} & \multicolumn{6}{c|}{\textbf{CodeLlama-7B}} & \multicolumn{6}{c}{\textbf{CodeLlama-13B}} \\
 & & \multicolumn{5}{c:}{\textbf{Temperature}} & \multirow{2}{*}{\textbf{Tot.}} & \multicolumn{5}{c:}{\textbf{Temperature}} & \multirow{2}{*}{\textbf{Tot.}}  & \multicolumn{5}{c:}{\textbf{Temperature}} & \multirow{2}{*}{\textbf{Tot.}}  \\
 & & 0.01 & 0.3 & 0.5 & 0.7 & 1.0 & & 0.01 & 0.3 & 0.5 & 0.7 & 1.0 & & 0.01 & 0.3 & 0.5 & 0.7 & 1.0 & \\
\midrule[1pt]
\multirow{3}{*}{\shortstack{\textbf{Equivalence}\\\textbf{No Delimit }\\\textbf{(10 programs)}}}
 & Unconstrained     & 0 & 0 & 0 & 0 & 0 & 0 & 0 & 0 & 0 & 0 & 0 & 0 & 0 & 0 & 0 & 0 & 0 & 0 \\
 & Grammar        & 0 & 0 & 0 & 0 & 0 & 0 & 0 & 0 & 0 & 0 & 0 & 0 & 0 & 0 & 0 & 0 & 0 & 0 \\
 & Semantic  & \textbf{7} & \textbf{7} & \textbf{8} & \textbf{8} & \textbf{3} & \textbf{33} & \textbf{8} & \textbf{9} & \textbf{9} & \textbf{9} & \textbf{8} & \textbf{43} & \textbf{10} & \textbf{8} & \textbf{8} & \textbf{10} & \textbf{6} & \textbf{42} \\[2pt]
\cdashline{1-20}
\multirow{3}{*}{\shortstack{\textbf{Equivalence}\\\textbf{Delimit}\\\textbf{(10 programs)}}}
 & Unconstrained     & 0 & 0 & 0 & 0 & 0 & 0 & 3 & \changed{4} & 2 & 3 & 1 & 13 & 4 & 2 & 2 & 3 & 2 & 13 \\
 & Grammar        & 0 & 0 & 0 & 0 & 1 & 1 & \changed{5} & \changed{6} & \changed{7} & 4 & 1 & 23 & \changed{6} & 5 & 2 & 2 & 0 & 15 \\
 & Semantic  & \textbf{10} & \textbf{10} & \textbf{8} & \textbf{9} & \textbf{8} & \textbf{45} & \textbf{9} & \textbf{10} & \textbf{10} & \textbf{6} & \textbf{6} & \textbf{41} & \textbf{9} & \textbf{8} & \textbf{7} & \textbf{8} & \textbf{6} & \textbf{38} \\
% \cdashline{1-20}
% \multirow{3}{*}{\shortstack{\textbf{Equivalence}\\\textbf{Rules + No Delimit}\\\textbf{(10 programs)}}}
%  & Unconstrained     & 0 & 0 & 0 & 0 & 0 & 0 & 0 & 0 & 0 & 0 & 0 & 0 & 0 & 0 & 0 & 0 & 0 & 0 \\
%  & Grammar        & 0 & 0 & 0 & 0 & 0 & 0 & 0 & 0 & 0 & 0 & 0 & 0 & 0 & 0 & 0 & 0 & 0 & 0 \\
%  & Semantic  & \textbf{4} & \textbf{3} & \textbf{4} & \textbf{6} & \textbf{3} & \textbf{20} & \textbf{8} & \textbf{6} & \textbf{10} & \textbf{6} & \textbf{4} & \textbf{34} & \textbf{9} & \textbf{10} & \textbf{10} & \textbf{7} & \textbf{6} & \textbf{42} \\[2pt]
% \cdashline{1-20}
% \multirow{3}{*}{\shortstack{\textbf{Equivalence}\\\textbf{Rules + Delimit}\\\textbf{(10 programs)}}}
%  & Unconstrained     & 0 & 0 & 0 & 0 & 0 & 0 & 3 & 3 & 5 & 0 & 1 & 12 & 2 & 2 & 4 & 2 & 0 & 10 \\
%  & Grammar        & 0 & 0 & 0 & 0 & 0 & 0 & 3 & 3 & 5 & 4 & 1 & 16 & 2 & 6 & 1 & 2 & 3 & 14 \\
%  & Semantic  & \textbf{8} & \textbf{8} & \textbf{7} & \textbf{8} & \textbf{6} & \textbf{37} & \textbf{8} & \textbf{7} & \textbf{5} & \textbf{8} & \textbf{8} & \textbf{36} & \textbf{8} & \textbf{8} & \textbf{7} & \textbf{9} & \textbf{7} & \textbf{39} \\
\midrule
\multirow{4}{*}{\shortstack{\textbf{TypeScript}\\\textbf{(72 programs)}}}
 & Unconstrained        & \changed{57} & \changed{55} & \changed{48} & \changed{45} & \changed{9} & \changed{214} & \changed{69} & \changed{71} & \changed{71} & \changed{66} & \changed{31} & \changed{308} & \changed{68} & \changed{67} & \changed{60} & \changed{47} & \changed{17} & \changed{259} \\
 & Grammar                 & \changed{38} & \changed{31} & \changed{28} & \changed{19} & \changed{10} & \changed{126} & \changed{67} & \changed{65} & \changed{62} & \changed{59} & \changed{36} & \changed{289} & \changed{68} & \changed{67} & \changed{57} & \changed{36} & \changed{8} & \changed{236} \\
 & Semantic             & \changed{\textbf{65}} & \changed{{62}} & \changed{{61}} & \changed{{58}} & \changed{{28}} & \changed{{274}} & \changed{\textbf{72}} & \changed{{71}} & \changed{\textbf{72}} & \changed{{66}} & \changed{{58}} & \changed{{339}} & \changed{\textbf{72}} & \changed{{69}} & \changed{{69}} & \changed{{66}} & \changed{{43}} & \changed{{319}} \\

\cdashline{3-20}
& \citet{typespaper}  & \changed{62} & \changed{\textbf{69}} & \changed{\textbf{71}} & \changed{\textbf{68}} & \changed{\textbf{67}} & \changed{\textbf{337}} & \changed{67} & \changed{\textbf{72}} & \changed{70} & \changed{\textbf{71}} & \changed{\textbf{71}} & \changed{\textbf{351}} & \changed{70} & \changed{\textbf{72}} & \changed{\textbf{70}} & \changed{\textbf{68}} & \changed{\textbf{67}} & \changed{\textbf{347}} \\
\bottomrule
\end{tabular}
}
\end{table}
\paragraph{Type-Safe Decoding}
% Unconstrained, DeepSeek-Coder-6.7b does not generate any TypeScript programs which compile. 
\changed{All three} models often produce compiling code, except at high temperatures, where their performance degrades.
\changed{The two CodeLlama models perform similarly, and the DeepSeek-Coder-6.7b model performs somewhat worse, especially at the highest temperature.}
% \changed{\sout{This degradation is much more severe for CodeLlama-13B:}} 
Surprisingly, CodeLlama-7B slightly outperforms CodeLlama-13B at \changed{most} temperatures.
When aggregated across all temperatures, semantic constrained decoding consistently outperforms the baselines, improving the performance of \changed{DeepSeek-Coder-6.7b}, CodeLlama-7B, and CodeLlama-13B by an average of \changed{28.0\%, 10.1\%, and 23.2\%} across all temperatures, respectively.
We also observe that grammar-constrained decoding alone actually performs worse than unconstrained generation. 
% \changed{\sout{on CodeLlama-7B and CodeLlama-13B}}. 
This effect is likely due to the fact that we limit the grammar to a small fragment of TypeScript and prevent programs that would compile, but are not in our supported fragment, from being generated.
\changed{Additionally, we observe that \citet{typespaper} performs slightly better than our semantic constrained decoding at medium temperatures and significantly better at the highest temperature. The reason for this discrepancy is that their tool supports a larger language fragment than we do, and the model tends to ignore the prompt and use forbidden language features as temperature increases.}

\subsubsection{Qualitative Evaluation}
\label{sec:qualitative-evaluation}
Next, we discuss some instances that show the strengths of semantic constrained decoding relative to GCD and unconstrained generation. 

\paragraph{Equivalence}
One interesting phenomenon is that while GCD ensures syntactic validity, it does not prevent models from creatively ``embedding'' invalid outputs within the grammar. 
For example, the text in \ref{fig:egraph_example_output_nat_language}
can be parsed (misleadingly) as a series of function calls due to the permissiveness of the grammar. 
This quirk allows LMs to ``escape'' into natural language.
By contrast, semantic constrained decoding disallows such outputs because programs with undefined variables cannot be equivalent to the original.
Next, we show a case where semantic constrained decoding prevents an inequivalent (but otherwise error-free) program from being generated.
CodeLlama-7B under GCD returns the code in \Cref{fig:egraph_example_output_inequiv} when given the program \lstinline{start + (end - start) * scale}.
This program needs to be extended with the suffix
% an additional term
\lstinline{+ start} to be correct. 
Semantic constrained decoding prevents the model from terminating the program string prematurely 
and would force it to continue (equivalence-constrained) generation.

\paragraph{TypeScript}
% The type errors we observed in the unconstrained and grammar-constrained mode were very simple, unlike the complex failure modes reported by \citet{typespaper} (e.g., failure to return along every control path, calling functions with too few parameters, etc.). Our pruner can prevent such errors, but we do not observe them in any generation mode. It is unclear whether this is because our TypeScript language fragment is smaller than theirs, because we did not evaluate against the HumanEval~\cite{} dataset like they did, or because we ran a different set of models.

\begin{figure}[t]
    \centering

    \begin{minipage}[t]{0.44\textwidth}
        \centering
\begin{subfigure}[t]{\textwidth}
\centering
\begin{verbatim}
Here is an equivalent refactored
version of the program using
explicit variable
\end{verbatim}
\caption{Natural language that parses as a valid program, generated with grammar-constrained decoding by CodeLlama-7B (see \Cref{sec:case-studies:egraphs}).}
\label{fig:egraph_example_output_nat_language}
\end{subfigure}

\vspace{4.6em}

\begin{subfigure}[t]{\textwidth}
\centering
\begin{verbatim}
let start_scaled = start * scale in
let end_scaled = end * scale in
(end_scaled - start_scaled)
\end{verbatim}
\caption{Inequivalent program generated with grammar-constrained decoding by CodeLlama-7B, for initial program \lstinline{start + (end - start) * scale} (see \Cref{sec:case-studies:egraphs}).}
\label{fig:egraph_example_output_inequiv}
\end{subfigure}
    \end{minipage}
    \hfill
    % Right column with one figure
        \begin{subfigure}[t]{0.52\textwidth}
            \centering 
\begin{verbatim}
function is_woodall(x: number): boolean {
  const divisors: number = x / 25;
  let countDivisors: number = 0;
  while (divisors > 1) {
      if (x % divisors === 0) {
          countDivisors += 1;
          break;
      } else {
          divisors = divisors - 1;
      }
  }
  ...
}
\end{verbatim}
  \caption{Snippet of a TypeScript program that mutates the constant variable \code{divisors}, generated during testing in grammar-constrained decoding mode by CodeLlama-7B at T=0.5 (see \Cref{sec:case-studies:types}).} \label{fig:typescript_immutable_variable}
        \end{subfigure}

    \caption{Sample incorrect outputs generated when not using semantic constrained decoding.}
    \label{fig:sample_outputs}
\end{figure}

Enforcing type safety prevents variable-name \changed{and function-name} hallucinations, eliminating many errors \changed{that appear in the GCD and unconstrained modes. We observe that  at high temperatures the unconstrained mode also struggles to correctly balance parentheses, and GCD often calls undefined functions when it wants to use forbidden language features (e.g., \code{newArray(...)} is used instead of \code{new Array} to try to construct an array).}
A more interesting \changed{(but rare)} case type-constrained decoding prevents is mutation of variables declared \lstinline{const}. 
Consider the code in \Cref{fig:typescript_immutable_variable}, produced with GCD.
The variable \lstinline{divisors} is declared constant on line 2, but it is mutated on line 9. 
In this case, the LM's mistake on line 2 does not cause an observable error until later in the program.
The type system implemented in our semantic pruner will prevent the assignment to \lstinline{divisors}, but it will still allow line 2.

When irresolvable conflicts between the LM's intention and the formal constraints of the decoder like \Cref{fig:typescript_immutable_variable} occur, the LM is typically unable to recover because it has already ``committed'' to a solution.
In the best case, this mistake results in continuations consisting of legal but useless code, such as whitespace or comments. In the absolute worst case, the LM will sometimes generate deeply nested arithmetic expressions that explode the size of our program space and slow our solver to the point of timing out.
In general, constrained decoding is known to skew the LM token distribution and how to mitigate this problem is an active area of research~\cite{gad}.
Such adjustments are orthogonal to our purposes, so we do not consider them.

\subsection{RQ2: Overhead}
\Cref{tab:realizability-timing} shows the average overhead of semantic constrained decoding (in ms) per generated token. 
Overhead on decoding time range from tens to a \changed{couple} hundred milliseconds per token, which is a very small price to pay for assurance provided by semantic constrained decoding.
As a reference, on our hardware, CodeLlama-13B takes on average 81 ms to produce a token when unconstrained.

\begin{table}[t]
\centering
\caption{Overhead of checking realizability in semantic constrained decoding (milliseconds/produced token).}
\label{tab:realizability-timing}
\resizebox{\textwidth}{!}{%
\begin{tabular}{@{}l|rrrrr|rrrrr|rrrrr@{}}
\toprule
\multirow{3}{*}{\shortstack{{Overhead in ms/token}}}
 & \multicolumn{5}{c|}{\textbf{DeepSeek-Coder-6.7b}} & \multicolumn{5}{c|}{\textbf{CodeLlama-7B}} & \multicolumn{5}{c}{\textbf{CodeLlama-13B}} \\
 & \multicolumn{5}{c|}{\textbf{Temperature}} & \multicolumn{5}{c|}{\textbf{Temperature}} & \multicolumn{5}{c}{\textbf{Temperature}} \\
 & 0.01 & 0.3 & 0.5 & 0.7 & 1.0 & 0.01 & 0.3 & 0.5 & 0.7 & 1.0 & 0.01 & 0.3 & 0.5 & 0.7 & 1.0 \\
\midrule
\textbf{Equiv No Delimit}  & 225 & 221 & 199 & 233 & 216 & 159 & 77 & 74 & 77 & 161 & 58 & 68 & 67 & 61 & 142 \\
% 224.79 & 221.01 & 199.33 & 232.84 & 215.92 & 159.40 & 77.36 & 74.17 & 77.36 & 161.43 & 58.48 & 67.66 & 66.96 & 61.03 & 141.89 \\
\textbf{Equiv Delimit}  & 82 & 73 & 564 & 198 & 162 & 50 & 55 & 65 & 121 & 156 & 63 & 74 & 79 & 49 & 128 \\
% & 81.62 & 72.98 & 563.77 & 198.23 & 161.64 & 50.24 & 55.31 & 65.25 & 121.35 & 155.73 & 63.29 & 73.68 & 78.77 & 48.5 & 127.97 \\
\textbf{TypeScript}
& \changed{202} & \changed{204} & \changed{185} & \changed{161} & \changed{167} 
& \changed{85} & \changed{64} & \changed{101} & \changed{107} & \changed{165}
& \changed{116} & \changed{107} & \changed{129} & \changed{135} & \changed{173} \\
% 363.76 & 347.03 & 342.73 & 289.70 & 235.62 &
% 239.54 & 207.19 & 252.65 & 222.68 & 356.33 &
% 300.99 & 227.87 & 236.05 & 273.75 & 323.46\\
\bottomrule
\end{tabular}
}
\end{table}

% We also measure how many LM tokens \name needs to prove realizability for before finding a realizable token.
\begin{figure}[t]
\centering
    \begin{subfigure}[b]{0.44\linewidth}
    \centering
        \includegraphics[width=0.97\linewidth]{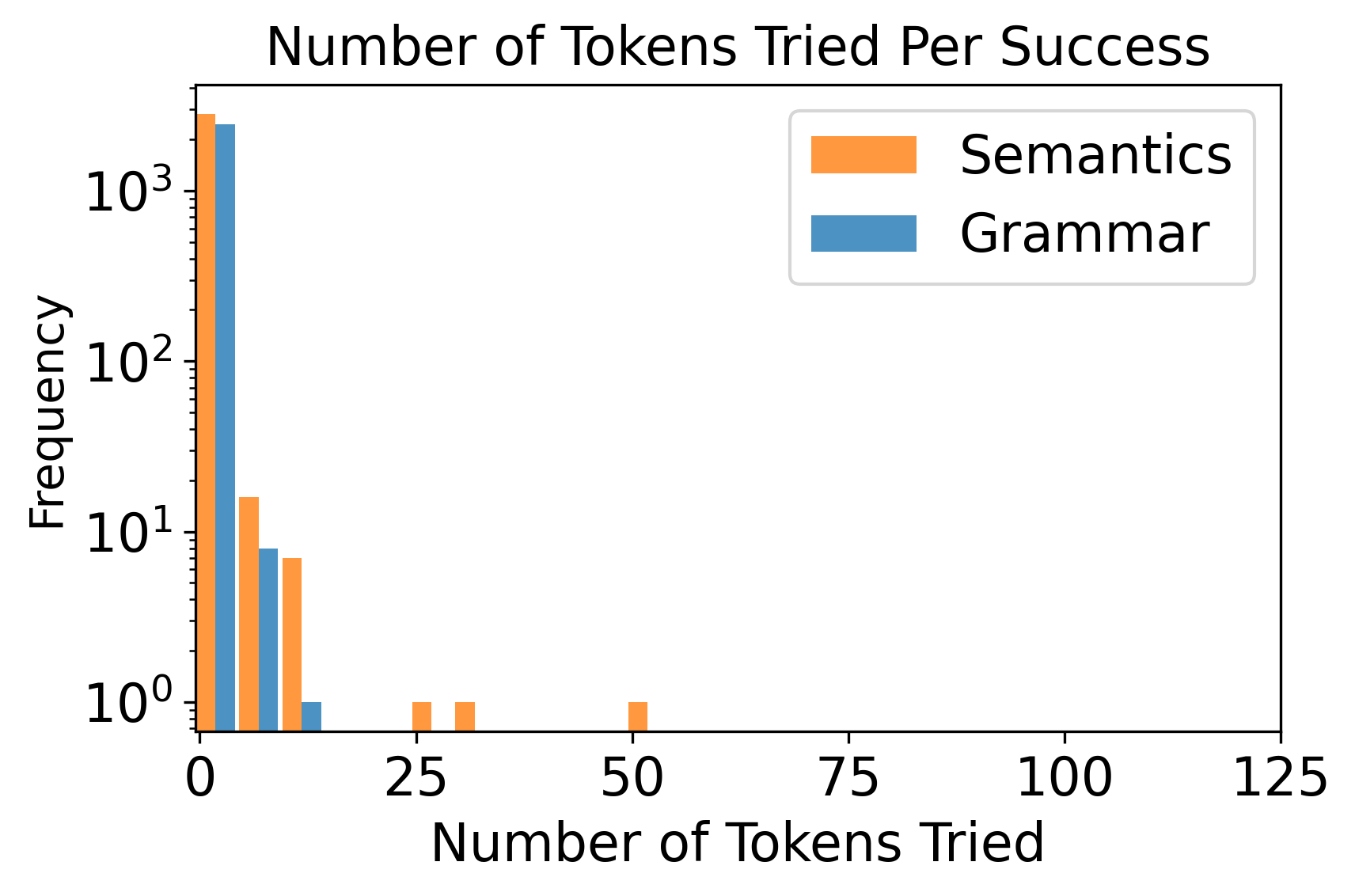}
        \caption{Equivalence, CodeLlama-7b, all temperatures}
    \end{subfigure}
    \hfill
    \begin{subfigure}[b]{0.44\linewidth}
    \centering
        \includegraphics[width=\linewidth]{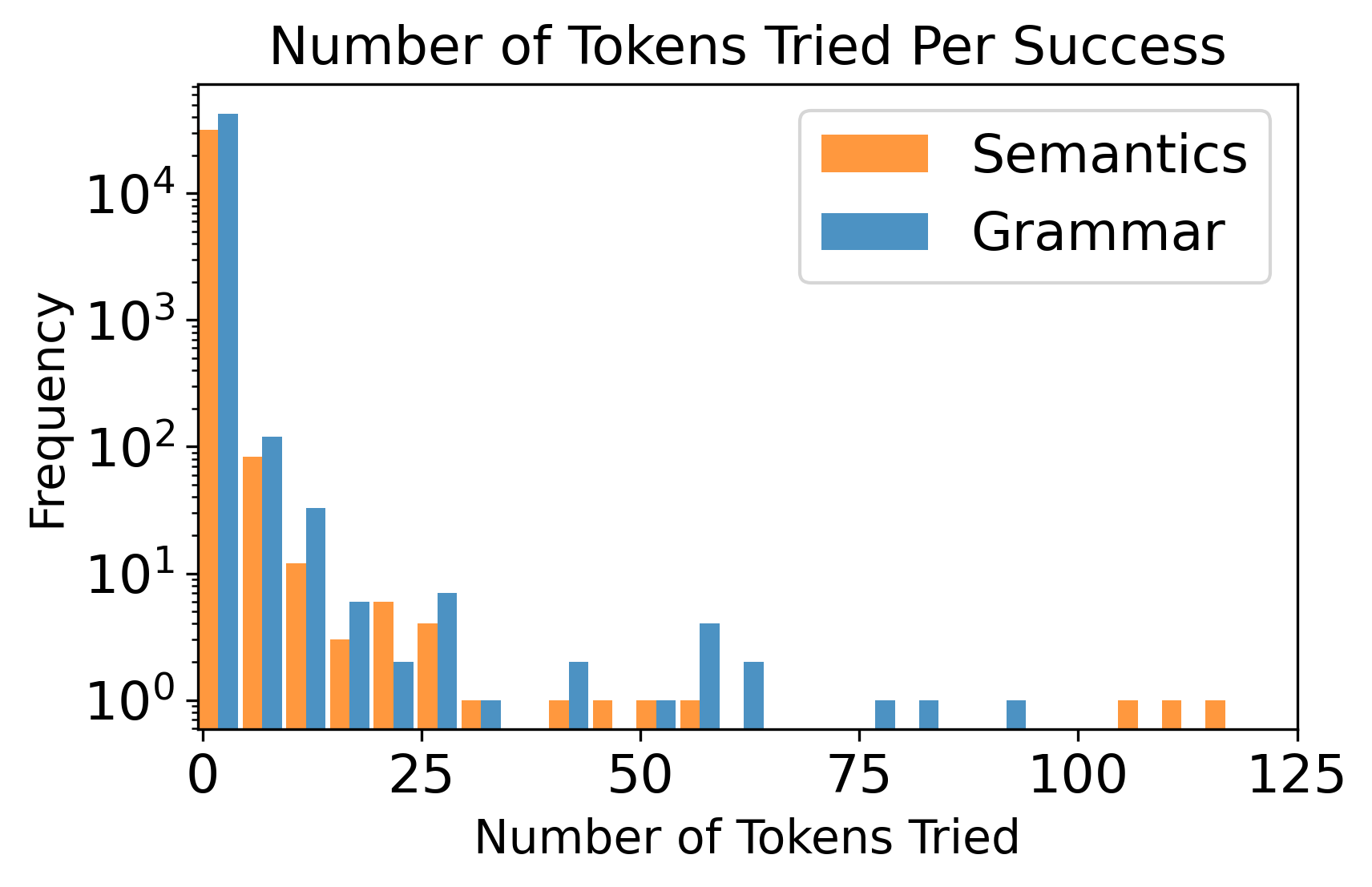}
        \caption{Typescript, CodeLlama-7b, all temperatures}
    \end{subfigure}
    \caption{Distribution of how many tokens were proven unrealizably by semantic constrained decoding to produce each individual token. The $k$th bar gives the number of successful tokens that were produced after trying between $5k$ and $5k+4$ unsuccessful tokens by CodeLlama-7b.}
    \label{fig:token_generation_charts}
\end{figure}

To better illustrate the source of the overhead, \Cref{fig:token_generation_charts} plots the distribution of the number of tokens that are tried before finding a realizable one, for CodeLLama-7B.
Results for other models can be found in \Cref{app:typescript_results}. 
In general, even with semantic constrained decoding the first token tried is accepted most of the time.
% , which is expected as typically LM only have high entropy for specific tokens that are particularly relevant to the final output.
For instance, across all models, \changed{96.4}\% of the tokens produced during semantic constrained decoding for TypeScript pass on the first try (similarly, \changed{96.1}\% of tokens for GCD).

% \subsection{Chat-style Interaction with Semantic Constrained Decoding}
% \label{sec:chat}

% \np{I think this should be in future work. We haven't done this, and it's very speculative.}
% \sn{Idk if we have time/space for this.}

% To address some of the limitations we mentioned, \name can also be used interactively. When the LM generates an incomplete or suboptimal solution, a user can issue corrective follow-up queries. Because semantic constraints are enforced at every step, follow-up generations stay within the allowed program space.

% While building a full interactive system around \name is outside the scope of this paper, we demonstrate this approach by re-querying a language model about an invalid e-graph generation. When prompted with a refinement (\texttt{"rewrite this term to avoid recomputing \loris{todo}}), the constrained decoder guides the model to a structurally valid, semantically equivalent variant that satisfies the new prompt.

% \loris{Add a walkthrough of this example, with prompt/output pairs and annotations}

\section{Related Work} \label{sec:related_work}
\paragraph{Constrained decoding.}
Constrained decoding techniques ensure the output of language models meets a given specification by throwing away invalid next tokens at each step.
Grammar-constrained decoding, in which the constraint is given as a context-free grammar, has been well studied \cite{syncode, geng2023grammarconstrained,willard2023efficient,wang2023grammar,greatgramma}.
A significant portion of \name's runtime is spent pruning tokens that fail simple syntactic validity. Integrating fast syntactic filtering tools such as LLGuidance~\cite{llguidance} as a pre-processing step could reduce this cost by eliminating invalid tokens early.
Beyond syntax, more recent work has explored enforcing specific semantic constraints such as type safety \cite{typespaper}. 
Frameworks such as monitors \cite{monitors} and completion engines \cite{synchromesh} provide abstractions for providing more complex constraints by allowing a user to provide a monitor (written in a general-purpose language) that performs the decoding.  
The main difference between our work and those techniques is that they require the user to write checkers over strings.
By contrast, \toolname operates at the level of abstract syntax, abstracting away the syntactic component and allowing the user to define pruners at the level of program spaces, thus enabling new applications such as equivalence-guided decoding.

\paragraph{Algebraic approaches to parsing.}
We build on a long line of work viewing parsers and grammars as algebraic, recursive structures \cite{typedAlgebraicParsing,kleene}.
\citet{parsingWithDerivatives} presents a functional approach to parsing based on applying Brzozowski derivatives to parser combinators.
Zipper-based variants such as ~\cite{parsingWithZippers, zippy} reduce redundant traversals in the basic version of parsing with derivatives to improve efficiency.
Integrating these techniques into our implementation could be another avenue to improve performance.
\paragraph{Regular coinduction.}
Our implementation relies on regular coinduction to represent and manipulate cyclic program spaces. CoCaml~\cite{cocaml,cocaml_foundation} is a framework to unambiguously define and compute recursive functions over regular codata. Because some recursive functions admit more than one interpretation on codata, the CoCaml language allows users to define custom solvers to implement their desired semantics. We do not use the CoCaml language directly. However, our Python backend handles the computation of corecursive functions, which produce codata (e.g., our pruners), with a solver analogous to the corec solver presented in \cite{cocaml}. Our backend's solver for \code{@fixpoint}-annotated functions, which compute concrete values over regular codata, is analogous to the fixpoint solver presented in \cite{cocaml}. 

% Regarding termination, we said it's hard to get syntactic condition.
% Cocaml people think so too.
% Syntactic conditions for well-behaved corecursion are generally tricky.
% Best analogue is guarded corecursion for productivity
% Productivity is ...
% Guarded corecursion is ...
% But even guarded corec is too weak.
% Moreover, neither of these are regularity, nor is our termination in our system.

\paragraph{Well-behaved corecursion.}
\changed{As discussed in \Cref{sec:coinduction-intro}, \toolname does not prevent users from writing corecursive functions that will not terminate.
Finding syntactic conditions for well-behaved corecursive functions is challenging in general.
One closely related body of work is on techniques for determining \textit{productivity} of corecursive functions (i.e., corecursive functions that are guaranteed to always eventually produce more of the infinite term they are computing~\cite{guarded_corecursion}). Famously, a sufficient condition for productivity is \emph{guarded corecursion}, which stipulates that every recursive call to a function must appear under a datatype constructor~\cite{guarded_corecursion}. This condition is often too restrictive, and a number of variants have been introduced to weaken it~\cite{beyond_guarded_corecursion1,beyond_guarded_corecursion2}.
In our case, productivity is necessary (but not sufficient) for regularity or termination of our CoCaml-like backend. For example, \code{f: int -> Stream} where \code{f(j) = j::f(j+1)} is a guarded function, and hence productive, but it is not regular and will not terminate in our backend.
}

\paragraph{Unrealizability and pruning in synthesis.}
Our approach draws inspiration from the concept of unrealizability---the problem of determining whether no solution exists that satisfies a given specification \cite{unrealizability2019}.
Existing approaches to proving unrealizability~\cite{unrealizability2019,unrealizability2020} typically focus on specific synthesis domains, leveraging domain insights to solve particular tasks.
Pruning, for example based on types~\cite{myth} or examples~\cite{sygus}, to remove infeasible portions of the search space is a well-established technique in program synthesis.
Our work provides a framework to adapt these general methods from traditional synthesis towards constraining the output of LLMs.

\changed{
\paragraph{Abstract interpretation.}
A natural question to ask is whether analyses that use abstract interpretation~\cite{abstractinterpretation} to detect invalid programs can be encoded as pruners in \toolname. In general, the complexity of the abstract domain affects how one can connect programs to abstract values.

When the set of abstract predicates is \textit{finite}, \citet{blaze} have shown how to connect programs in a search space with their corresponding abstract values using abstract finite tree automata (AFTA). In an AFTA, each state represents a predicate or abstract value, and the trees derivable from that state correspond exactly to the programs to which the abstract interpretation will assign the state’s abstract value. We can encode an AFTA by defining, for each state, a pruning function that encodes the transitions to that state. The example we used in Section~2 involving odd/even predicates is already a (simple) instantiation of such an encoding.

For more complex \textit{infinite} abstract domains (e.g., numerical ones), assigning abstract values to parts of the program space becomes harder. Prior work in synthesis and unrealizability~\cite{unrealizability2020} shows that grammar-flow analysis (GFA)~\cite{gfa} provides a systematic solution to do so. GFA analyzes programs in a regular tree grammar (in our case, the program space) by performing dataflow analysis and assigning abstract values from a semilattice to each nonterminal. In a GFA, abstract transformers assigned to individual operations are lifted to grammars by interpreting the unions of productions as joins in the semilattice. By solving the “abstract” equations resulting from these interpretations, one can obtain a representation of the program space that matches the given abstract interpretation. While solving these equations can be challenging, prior work has identified cases where doing so is feasible~\cite{monotonic}.
}
\paragraph{Bidirectional Typing}
In a bidirectional typesystem \cite{bidirectional}, type checking and type inference are done simultaneously. As noted in \Cref{sec:case-studies:types}, this style of typechecking can be naturally expressed as a pruner in \toolname. Type checking corresponds to describing how to build, via a corecursive function, a well-typed space of programs from well-typed component subspaces. Type inference corresponds to performing an analysis on a space of programs via a fixpoint computation.

\section{Conclusion} \label{sec:conclusion}

We introduced \name, a new framework for semantic constrained decoding that 
allows one to impose semantic constraints directly on the abstract syntax trees representing programs (instead of their string syntax).
\name lets users
% allows one to 
program constraints by providing \textit{semantic pruners}---recursive program operating over finite representations of the infinitely many programs the LM can produce on a given prefix.
This flexibility enables new applications---e.g., constraining LMs to only output programs that are equivalent (up to rewrite rules) to a given input program.

\name bridges the gap between formal methods and the outputs of language models. Looking forward, there are numerous opportunities to enhance the framework. 
\changed{In terms of programmability, users of \toolname must directly define a pruner for the property they wish to enforce as a corecursive function on program spaces.
Moreover, this function must return a space that is regular.
There may be a better interface for programming pruners that abstracts away these details.
In terms of performance, the current implementation of \toolname is relatively unoptimized, incurring noticeable per-token overhead: future work could likely greatly improve the runtime performance.
}
Exploring new backends for reasoning about codata, such as term rewriting solvers~\cite{egg,egglog} or constrained Horn clauses~\cite{chc}, could enable new applications. Finally, we envision better combinators for expressing semantic pruners and for automatically proving whether pruners over-approximate or under-approximate a user-provided constraint.

% \loris{something smart we can say given how much more we know now?}

% \loris{beyond regular codata}

% \loris{an algebra of constraints}

% \loris{when is it decidable?}

% There are many potential optimizations that could be applied to reduce \name's overhead.
% \Cref{fig:token_generation} shows that a lot of \name's effort towards proving unrealizability is spent towards proving that some tokens are not syntactically valid.
% To process a token, \name uses its corecursive implementation of parsing with derivatives, which is much slower than existing GCD tools such as LLGuidance~\cite{llguidance}.
% Such a tool could be queried first in a pre-processing step so that \name does not have to check tokens which are not syntactically valid.
% More sophisticated versions of derivative-based parsing such as \cite{parsingWithZippers} and \cite{zippy} have been demonstrated to improve performance greatly: implementing these would likely also reduce overhead. 

% \sn{Some other possible future work:}
% \sn{Maxing cocaml fixpoint solver allow more permissively typed-functions, perhaps unifying fixpoint and corec solving.}

% \sn{Optimizing cocaml solvers in general}

% \sn{Translate attribute grammar -> parser}

% \sn{Formalize the intuition that bidirectional typesystems can be translated into pruners.
% More broadly, can we interpret \textit{arbitrary} recursive functions as performing some kind of checking/inference and lift them to the appropriate function on \code{ProgramSpace}?}

% \sn{Formalize these ideas with functors?
% Not sure that's appropriate to mention here though.}

\begin{acks}
The project is supported in part by a Microsoft Faculty Fellowship; a UCSD JSOE Scholarship; and NSF under grants \grantnum{00003}{CCF-2422214}, \grantnum{00003}{CCF-2506134}, \grantnum{00003}{CCF-2446711}. This material is based upon work supported by the National Science Foundation Graduate Research Fellowship Program under Grant No. \grantnum{00003}{DGE-2038238}.
Any opinions, findings, and conclusions or recommendations expressed in this publication are those of the authors, and do not necessarily reflect the views of the sponsoring entities.
Loris D’Antoni holds concurrent appointments as a Professor at the University of California San Diego and as an Amazon Scholar. This paper describes work performed at the University of California San Diego and is not associated with Amazon.
\end{acks}

\section*{Data-Availability Statement}
The code for this paper is available on \href{https://github.com/large-loris-models/chopchop}{GitHub}.
An artifact containing a Docker image is available on Zenodo~\cite{artifact}.

\bibliography{main}
\newpage
\appendix
\section{Appendix}
\subsection{Soundness and Completeness Proofs}\label{app:proofs}
\changed{We prove the main results of \Cref{sec:soundness-completeness}. We first prove \Cref{lem:cd_sound}:

\begin{proof}
    Let $\llm$, $\stringConstraint$, and $\realchecker_\stringConstraint$ be given. Suppose \Cref{alg:constrained-decoding} returns a string $\omega' \in \alphabet\st$ which does not satisfy $\stringConstraint$. The only way for $\omega'$ to be returned is from Line 7. But we check that the output of Line 7 satisfied $\stringConstraint$ on Line 6! By contradiction, no string $\omega' \in \alphabet\st$ can be returned which does not satisfy $\stringConstraint$.
\end{proof}

We next prove \Cref{lem:cd_complete}. The proof is identical for top-p and top-k sampling:
\begin{proof}
    Suppose $\omega$ is a valid solution that can be found with non-zero probability in unconstrained (top-p/top-k) decoding. Consider an execution of unconstrained decoding that occurs with non-zero probability and returns a solution $\omega$ that satisfies $\stringConstraint$. Such an execution is guaranteed to exist because there are only finitely many executions that can return $\omega$. Let $\alpha = (\alpha_1, \dots, \alpha_n = \omega\endtok)$ be the sequence of elements dequeued during the execution. Since the search is non-backtracking (enqueueing $\omega\tau$ deletes all elements that don't start with $\omega$), every $\alpha_j$ is a prefix of $\alpha_{j+1}$. In particular, every $\alpha_j$ for $j < n$ is realizable. If $\realchecker_\stringConstraint$ is complete, then this exact execution is also possible in the constrained case with the same probability.

    The probability of producing $\omega$ in the unconstrained case is the sum of the probabilities of producing each of the (finitely many) executions that produce $\omega$. Since constrained decoding can produce every execution that the unconstrained case can with the same probability, the probability that unconstrained decoding produces $\omega$ is at least as large as the probability that unconstrained decoding produces $\omega$.
\end{proof}

Note that \Cref{lem:cd_complete} also applies to greedy decoding (i.e., top-k sampling for k=1).
}
\subsection{Lexing} \label{app:lexing}
We open with a brief background on maximal munch lexing, the most widely used lexing formalism.

\paragraph{Maximal Munch Lexing}
A maximal munch lexer~\cite{compilers_dragon_book} is instantiated by a collection of disjoint regular expressions, each of which corresponds to a different kind of lexeme. For example, we might give the regex \code{[1-9][0-9]*} to describe the set of strings representing integers, and the regexes \code{true} and \code{false} to describe the strings encoding keywords true and false, respectively. We call these regexes \emph{lexeme classes}.

Given such a collection of disjoint regexes, a \emph{lex} of a string $\omega$ is a partition $(\omega_1, \dots, \omega_n)$ of $\omega$ so that each $\omega_j$ matches one of the given regexes. These strings $\omega_j$ are called \emph{lexemes}. The \emph{maximal munch lex} of $\omega$ is the unique lex so that, for any other lex $(\omega_1, \dots, \omega_j, \omega_{j+1}', \dots \omega_m')$, we have $\abs{\omega_j} \geq \abs{\omega_{j+1}'}$.

\paragraph{Lexing a Partial Program}
Given a string $\omega$ that represents a partial program, our goal is to construct a representation of set of maximal munch lexes of all the strings that extend $\omega$. This will be the output we pass to the parser. For example, the string \code{1 + 3} can be extended to \code{1 + 34}, whose maximal munch lex is \code{[1, +, 34]}, or \code{1 + 3 + 4}, whose maximal munch lex is \code{[1, +, 3, +, 4]}. To represent the set of maximal munch lexes of completions of $\omega$, we will build a \emph{partial lex} $L$ like the following:
$$L = \{\mcode{[1, +, 3]}, \mcode{[1, +, 3[0-9]+]}\}$$
Each element of $L$ is a \emph{lexical prefix} -- a sequence of lexemes that ends in a regex. A lexical prefix describes the set of prefixes that match it. So
\code{[1, +, 3]} describes the lexes \code{[1, +, 3, +]}, \code{[1, +, 3, 5, 6]}, etc. And \code{[1, +, 3[0-9]+]} describes the lexes \code{[1, +, 31, +]}, \code{[1, +, 354, 6]}, etc. The partial lex $L$ describes the set of lexes that extend any one of its lexical prefixes.
Our goal will be to produce $L$ from $\omega$ so that the lexes described by $L$ are exactly the longset match lexes of completions of $\omega$.

To build $L$ incrementally, we introduce a richer representation of $L$ that tells us how much of a regex has already been matched by $\omega$. For example, if I had a ``print'' lexeme, then $$L = \{\mcode{[print]}\}$$ could be the partial lex of both $\omega = \text{``print''}$ and $\omega = \text{``pri''}$. To resolve this, we add a $@$ annotation that tells us how much of a regex has been matched. Then, I can distinguish \code{[print @]} from \code{[pri @ nt]} Left of the $@$ is a string that has been explicitly matched by the end of $\omega$, and right of the $@$ is a regex that has not been matched yet.

To advance an annotated regex by a character $a$, we define $D_a(\mcode{x @ y}) = \mcode{xa @ }D_a(\mcode{y})$
, where $D_a(\mcode{y})$ is the usual Brzozowski derivative of \code{y} with respect to $a$~\cite{Brzozowski}. For example,
$$D_i(\mcode{pr @ int}) = \mcode{pri @ nt}$$

Given $\omega$, we will incrementally build a \textit{lexer state}, a set of sequences of such annotated regexes \code{x @ y} like
$$\lexerstate = \{(\mcode{pri @ nt})\}$$
that projects to the desired $L$ when annotation symbols $@$ are removed.
% For a leaf sort $\tau$ denoting a set of predicates on lexemes/AST leaves, we denote $B_\tau$ as the set of singleton predicates of $\tau$ and $B = \bigcup B_\tau$. Thus, $B_\intSort = \Z$.

% Formally, each sequence in a lexer state is a \textit{lexical prefix}, a sequence $$l = (x_1 @ \epsilon, \dots, x_{n-1} @ \epsilon, x_n @ y_n)$$ A lexical prefix $l$ encodes $\semlexerstateprog{l}$, the set of lexes whose prefixes match its predicates. Thus
% $$\semlexerstateprog{(\termfont{pri @ nt})} = \{(\termfont{print}), (\termfont{print}, \termfont{17}), (\termfont{print}, \termfont{+}, \termfont{+}, \termfont{print}), \dots\}$$

% For partial lexes,
% $$\semlexerstateprog{\lexerstate} = \bigcup_{l \in \lexerstate} \semlexerstateprog{l}$$
% and $\semlexerstateprog{L}$ is defined analogously for unannotated partial lexes $L$.

% \begin{definition} [Lexer State]
%     A \textit{lexical prefix} is a sequence $t = (t_1, \cdots, t_n)$ of annotated regexes (with associated sorts) in which the first $n-1$ entries $t_j$ have the form $x_j @ \epsilon$, and $t_n$ has the form $x_n @ y_n$. A lexical prefix encodes a set of lexes $\semlexerstateprog{t} = \{(x_1', \cdots, x_m') \mid \forall j \leq n. x'_j \in t_j\}$.

%     A \textit{lexer state} $\lexerstate$ is a finite set of lexical prefixes. A lexer state $\lexerstate$ encodes a set of lexes given by $\semlexerstateprog{\lexerstate} = \bigcup_{p \in \lexerstate} \semlexerstateprog{p}$.
% \end{definition}
% For unannotated partial lexes $L$, we define $\semlexerstateprog{L}$ similarly.

The full algorithm to build $L$ is given in \Cref{alg:lexing}. It very closely mirrors the na\"ive approach to maximal munch lexing~\cite{compilers_dragon_book}. We iterate through $\omega$ character by character, constructing the lexer state $\lexerstate$ incrementally. As we go, we discard the partial lexes from $\lexerstate$ that fail maximal munch by not having the largest possible tokens from left to right. At the end of the algorithm, we remove the pointers and throw away ``ignorable'' tokens (e.g., whitespace and comments) to convert $\lexerstate$ into $L$.

\begin{figure*}
\begin{algorithm}[H]
\caption{Partial Lexing ($\lexset$)}\label{alg:lexing}
\begin{algorithmic}[1]
\Procedure{$\lexset$}{$\omega \in \Sigma^*$}
    \State $\lexerstate = \growlexerstate(\omega)$
    \State $L' = \rmvpointers(\lexerstate)$
    \State $L = \rmvignores(L')$
    \State \Return $L$
\EndProcedure
\State
\State \text{@memoize}
\Procedure{$\growlexerstate$}{$a_1 \dots a_n \in \Sigma^*$}
    \If{$n = 0$}
        \State \Return $\{()\}$
    \Else
        \State $\lexerstate \gets \growlexerstate(a_1 \dots a_{n-1})$
        \State $\lexerstate \gets \extendlexes(\lexerstate, a_n)$
        \State $\lexerstate \gets \munch(\lexerstate)$
        \State \Return $\lexerstate$
    \EndIf
\EndProcedure
\State
\Procedure{$\extendlexes$}{$\lexerstate : LEXERSTATE$, $a \in \Sigma$}
    \State $\text{result} \gets \emptyset$
    \For{each $l = [l_1, \cdots, l_m] \in \lexerstate$}
        \If{$m = 0$}
            \State $\text{result} \gets \text{result} \cup \{(D_a(@ \top_\tau)) \mid D_a(\top_\tau) \neq \bot_\tau\}$
        \EndIf
        \If{$\epsilon \in y_m$}
        % \Comment{If the current token can be completed, move to the next}
            \State $\text{result} \gets \text{result} \cup \{[l_1, \cdots, l_{m-1}, x_m @ \epsilon, D_a(@ \top_\tau)] \mid D_a(\top_\tau) \neq \bot_\tau\}$
        \EndIf
        \If{$D_a(l_m) \neq \bot_\tau$}
        % \Comment{Extend the current token if possible}
            \State $\text{result} \gets \text{result} \cup \{[l_1, \cdots, l_{m-1}, D_a(l_m)]\}$
        \EndIf
    \EndFor
    \State \Return $\text{result}$
\EndProcedure
\State
\Procedure{$\munch$}{$\lexerstate : LEXERSTATE$}
    \For{each $l = [l_1, \cdots, l_m] \in \lexerstate$}
        \If{$\exists l' = [l_1, \dots, l_{k-1}, l_k', \dots, l'_{m'}] \in \lexerstate. ~(\abs{x_k} < \abs{x_k'} \land \epsilon \in y_k')$ }
            \State $\lexerstate.pop(l)$
        \EndIf
    \EndFor
    \State \Return $\lexerstate$
\EndProcedure
\end{algorithmic}
\end{algorithm}
\caption{The function $\growlexerstate$ produces a lexer state $\lexerstate$ for $\omega$ by iteratively extending partial lexes by the next character ($\extendlexes$) and discarding partial lexes which fail maximal munch by not having the largest possible tokens from left to right ($\munch$). The outermost function $\lexset$ turns $\lexerstate$ into $L$. Above, $\Sigma$ represents an alphabet of characters. We memoize $\growlexerstate$ so that, when a prefix $\omega$ grows to $\omega\alpha$, computing $\growlexerstate(\omega\alpha)$ will reuse the earlier computation of $\growlexerstate(\omega)$.}
\end{figure*}

If there is a reserved whitespace character (e.g., \text{' '}) and a class of lexemes that matches a single occurrence of that character (e.g., \code{\s+}), then the lexer state $L$ that our algorithm produces describes exactly the set of maximal munch lexes of completions of $\omega$, up to the presence/absence of ignorable tokens like whitespace.

\begin{theorem} [Correctness of Partial Lexing] \label{thm:lexing_is_correct}
    Suppose every string in $\alphabet\st$ is expressible by some token sequence in $\tokens\st$.
    Suppose the regexes that each language token matches are disjoint.
    Suppose $\lex$ is a maximal munch lexer.
    Suppose there is an whitespace character in our language which matches the ignorable lexeme class and appear in no lexeme of any other class.
    Then for every LM token sequence $\omega \in \tokens\st$, the set of sequences of language token predicates $L = \lexset(\omega)$ corresponds exactly to the maximal munch lexes of continuations of $\omega$ with whitespace tokens removed (i.e., $\semlexerstateprog{L} = \{\lex(\omega\beta) \mid \beta \in \tokens^*\}$).
\end{theorem}

\begin{proof}
    We first prove two properties of $\growlexerstate$.
    \begin{lemma} [Property 1]
        The concatenation of the left-hand sides of every lexical prefix in $\growlexerstate(\omega)$ is exactly $\omega$.
    \end{lemma}
    \begin{proof}
        The proof proceeds via straightforward induction on the length of $\omega$.
        When $\omega = \epsilon$, we have $\tilde{L} = \{()\}$.
        When $\omega' = \omega a$, we know that every lexical prefix in $\tilde{L}$ on line 12 matches left-hand-sides with $\omega$.
        Each partial lex added by $\extendlexes$ takes a derivative in $a$ of a lexical prefix from the input $\tilde{L}$.
        So, each partial lex in the output matches left-hand-sides with $\omega'$.
    \end{proof}

    \begin{lemma} [Property 2]
        For every $\omega$ and every $\omega' \succeq \omega$, the maximal munch lex of $\omega'$ is matched by exactly one lexical prefix in $\growlexerstate(\omega)$, including whitespace.
    \end{lemma}
    \begin{proof}
        The second property is proven by induction on the length of $\omega$ as well.
        When $\omega = \epsilon$, we have $\tilde{L} = \{()\}$ and the lex of any continuation of $\omega$ matches $()$.
        When $\omega' = \omega a$, some prefix of length $m$ in $\tilde{L}$ on line 12 matches the lex of $\omega'$.
        If the $m$th token of $\omega'$ matches the lhs of the $m$th token of the lexical prefix exactly, then line 21 or line 23 adds a prefix $l$ with $m+1$ tokens that matches $\omega'$.
        Similarly, if the $m$th token of $\omega'$ is longer than the lhs of the $m$th token of the lexical prefix exactly, then line 25 adds a prefix $l$ with $m$ tokens that matches $\omega'$.
        Regardless, this prefix $l$ cannot be removed by $\munch$.
        If $l$ were removed, this would imply that one can write some other lexical prefix $l'$ whose left-hand-sides match $\omega'$ (by Property 1) but which is a maximal munch of $l$. But this would imply that our lex of $\omega'$ is not a maximal munch lex of $\omega'$ since there exists a better tokenization, a contradiction.
        So $\omega'$ matches at least one lexical prefix in $\tilde{L}$.
        Moreover, $\omega'$ matches at most one lexical prefix in $\tilde{L}$ because, for any two matches of $\omega'$, one would necessarily munch the other.
    \end{proof}

    The final result follows from Properties 1 and 2.
    When we remove the annotations and whitespace from whichever lexical prefix matches the whitespace-containing $\omega' \succeq \omega$, we get a lexical prefix matching the whitespace-lacking lex of $\omega'$.
    So $\semlexerstateprog{L} \supseteq \{\lex(\omega\beta) \mid \beta \in \tokens^*\}$.
    For the other direction, let $l =[\pltoken_1,\dots,\pltoken_n]$ be a lex that matches some lexical prefix in $\semlexerstateprog{L}$.
    Since that lexical prefix has left-hand-sides matching $\omega$, we know that $\omega$ must match a prefix of $l$ up to whitespace.
    Suppose $$\omega = \pltoken_1 \;\; \pltoken_2\pltoken_3 \;\; \pltoken_4 \dots \pltoken_{m-1} \;\; a$$ where $ab = \pltoken_m$.
    Then we can construct a completion of $\omega$ that lexes to $l$ by simply putting whitespace between each of the tokens following $\pltoken_m$: $$\omega' = \omega b\;\;\pltoken_{m+1}\;\;\pltoken_{m+2}$$
    So $\semlexerstateprog{L} \subseteq \{\lex(\omega\beta) \mid \beta \in \tokens^*\}$.
\end{proof}

\subsection{TypeScript Typechecking}
\subsubsection{Details of Type Pruning}\label{app:typescript}
\changed{The main idea of our typepruner is presented in \Cref{sec:case-studies:types}. Here, we provide our grammar, selected inference rules for single programs, and selected parts of the implementation of \typerestrict. We also provide some additional information on our use of bounded types.}

We present our typescript grammar in \Cref{fig:typescript_grammar}. Typing rules infering the types of individual TypeScript programs are presented in \Cref{fig:typescript_single_program_inference_rules_expressions,fig:typescript_single_program_inference_rules_statements}. We use a terse, inference-only typesystem for individual programs that is used when pruning sets of programs. When we write, e.g., $\Gamma \vdash e \implies \texttt{bool}$, we mean that $e$ infers to a type that matches \texttt{bool}. Similarly, when we say $\Gamma \vdash s \implies \tau - ()$, we mean that $s$ infers to a subtype of $\tau$ that does not contain the unit type. Note that our typing contexts $\Gamma$ assign types and a designation of mutable/immutable to variables.
A much cleaner type system is given in \citet{typescript_good_typesystem}, but this one suffices for our purposes. We borrow from \citet{typescript_good_typesystem} the notion that statements type to their return types.

\changed{Part of the implementation of \typerestrict is given in \Cref{fig:type_pruner-example-implementation}. As mentioned in \Cref{sec:case-studies:types}, most cases like \code{Add} and \code{Ternary} can be handled by pruning and combining subterms. For cases like \code{Apply}, we prune \code{f} to have type \code{? -> tau} until \code{f} collapses to a single AST. Once \code{f} collapses, we can infer its type and typeprune \code{x} appropriately.}

In \Cref{sec:case-studies:types}, we promised to explain our use of bounded types. Recall the function application rule from \Cref{sec:case-studies:types}. There is a flaw in the behavior of that rule if we allow types of arbitrary size. If \code{f} unrolls to an application of itself (\code{Apply (f _)}), then the coterm that \code{(type_prune gamma f (T -> tau))} represents may not be regular. In order to compute 
\begin{lstlisting}
(type_prune gamma f (T -> tau))
\end{lstlisting}
we must eventually type-prune a case of the form \code{Apply ( f a )} against the type \code{(T -> tau)}. If \code{f} is incomplete, then our definition of \typerestrict, asks us to compute 
\begin{lstlisting}
(type_prune gamma f (T -> T -> tau)).
\end{lstlisting}
This process continues infinitely, generating larger and larger types. Because infinitely many distinct calls to \typerestrict are made, we cannot represent the result as a regular coterm.

One solution is to simply return \code{(Apply f x)} when \code{f} is incomplete. This would mean that we do not typecheck a function until it is complete. For the purposes of constrained decoding, this is unsatisfactory; we ought to put some kind of restriction on the LM's behavior as it generates the left-hand-sides of function applications. Worse, since the TypeScript syntax does not tell us whether an expression is the left-hand side of a function application until it is applied, many otherwise invalid expressions can fail to typecheck if we leave function calls uninspected. Instead, we simply bound the our types to depth 3. In our example above, the type \code{T -> T -> T -> tau} becomes \code{T -> T -> T}. This prevents \typerestrict from constructing infinitely many types and makes our analysis regular. However, we do sacrifice soundness for programs with deeply nested types (i.e., \typerestrict may accept a singleton set with an ill-typed program if trying to typecheck the program generates deeply-nested types). We did not find any such large types appearing in the LLM generated programs in practice.

\changed{The challenges in handling function application also explain why we implemented an overapproximate pruner, rather than an underapproximate pruner. For an underapproximate pruner, the obvious sound way to handle \code{(Apply f x)} is to consider a finite subset of the possible combinations of types (e.g., by only considering types up to some size bound, or inferring the possible types of \code{f}). While we do bound the size of our types to prevent nonterminating behavior, pruning \code{f} and \code{x} for each valid combination of types is expensive. More than that, it is unnecessary. Since LMs generate programs from left to right, \code{f} will always resolve to a concrete function before any of \code{x} is written. This means that we will always be able to infer the type of \code{f} when we start writing \code{x}. As a result, the only overapproximation that matters to constrained decoding is bounding the size of types. As explained above, we overapproximate types with depth greater than 3 (e.g., \code{((int -> ...) -> int) -> int}) by types of depth 3 (e.g., \code{((T -> int) -> int) -> int}). We did not find any such large types appearing in the generated programs in practice.}
\newpage
\begin{figure}
    \centering
\begin{lstlisting}
type_prune :: Environment -> Type -> ProgramSpace -> ProgramSpace

type_prune env tau Empty =
    Empty

type_prune env tau (Union [a1, a2, ...]) =
    Union [(type_prune env tau a1), (type_prune env tau a1), ...]

type_prune env tau (Int x) =
    if int <= tau then (Int x) else Empty
    
type_prune env tau (Add e1 e2) =
    if int <= tau
    then Add (type_prune env e1 Int)
              (type_prune env e2 Int)
    else Empty
    
type_prune env tau (Ternary b s1 s2) =
    Ternary (type_prune env bool b)
        (type_prune env tau s1)
        (type_prune env tau s2)
        
type_prune env tau (Apply f x) = 
    case (collapse f) of
        Some s ->
            let tau1 -> tau2 = (type_infer env s)
            in 
              if tau2 <= tau 
              then Apply f (type_prune env argType x)
              else Empty
        Nothing ->
            Apply (type_prune env (T -> tau) f) x
...

typepruner = (type_prune {} T)
\end{lstlisting}
    \caption{\changed{Some selected rules for the corecrusive \typerestrict function. Given a typing environment, a target type, and a program space, \typerestrict overapproximates the set of programs in the program space that typecheck to the target type under the typing environment. The top-level typepruner \code{typepruner} is implemented in terms of \typerestrict.}}
    \label{fig:type_pruner-example-implementation}
\end{figure}

\begin{figure}[ht]
\centering
\begin{minipage}[t]{0.45\textwidth}
\centering
\textbf{Statement Grammar} \\
\[
\begin{array}{ll}

\text{Statements} \to & \\ 
                         & \text{\hspace{-7em}Statement} \\
                         & \text{\hspace{-7em}| Statement Statements} \\

\\
\text{Statement}  \to & \\ 
                    & \text{\hspace{-7em}Assignment ;} \\
                    & \text{\hspace{-7em}| Exp ;} \\
                    & \text{\hspace{-7em}| return Exp ;} \\
                    & \text{\hspace{-7em}| Block} \\
                    & \text{\hspace{-7em}| function Var () : Type Block} \\
                    & \text{\hspace{-7em}| function Var ( Params ) : Type Block} \\
                    & \text{\hspace{-7em}| for ( Assignment ; Exp ; Reassignment ) Block} \\
                    & \text{\hspace{-7em}| while ( Exp ) Block} \\
                    & \text{\hspace{-7em}| if ( Exp ) Statement else Statement} \\
                    & \text{\hspace{-7em}| if ( Exp ) Statement} \\

\\
\text{Assignment}  \to & \\ 
                       & \text{\hspace{-7em}let Var : Type = Exp} \\
                       & \text{\hspace{-7em}| const Var : Type = Exp} \\
                       & \text{\hspace{-7em}| let Var = Exp} \\
                       & \text{\hspace{-7em}| const Var = Exp} \\
                       & \text{\hspace{-7em}| Reassignment} \\

\\
\text{Reassignment}  \to & \\ 
                         & \text{\hspace{-7em}Var = Exp} \\
                         & \text{\hspace{-7em}| Var ++} \\
                         & \text{\hspace{-7em}| ++ Var} \\
                         & \text{\hspace{-7em}| Var += Exp} \\

\\
\text{Type}  \to & \\ 
                 & \text{\hspace{-7em}number} \\
                 & \text{\hspace{-7em}| boolean} \\
                 & \text{\hspace{-7em}| () => Type} \\
                 & \text{\hspace{-7em}| ( Params ) => Type} \\

\\
\text{Block}  \to & \\ 
                  & \text{\hspace{-7em}\{\}} \\
                  & \text{\hspace{-7em}| \{ Statements \}} \\

\end{array}
\]
\end{minipage}
\hfill
\begin{minipage}[t]{0.45\textwidth}
\centering
\textbf{Expression Grammar} \\
\[
\begin{array}{ll}

\text{Exp}  \to & \\ 
                & \text{\hspace{-16em}P3\_Exp} \\
                & \text{\hspace{-16em}| P3\_Exp ? Exp : Exp} \\

\\
\text{P3\_Exp} \to & \\ 
                             & \text{\hspace{-16em}P2\_Exp} \\
                             & \text{\hspace{-16em}| P2\_Exp \&\& P3\_Exp} \\
                             & \text{\hspace{-16em}| P2\_Exp || P3\_Exp} \\

\\
\text{P2\_Exp}  \to & \\ 
                              & \text{\hspace{-16em}P1\_Exp} \\
                              & \text{\hspace{-16em}| P1\_Exp < P1\_Exp} \\
                              & \text{\hspace{-16em}| P1\_Exp == P1\_Exp} \\
                              & \text{\hspace{-16em}...} \\

\\
\text{P1\_Exp}  \to & \\ 
                              & \text{\hspace{-16em}Base\_Exp} \\
                              & \text{\hspace{-16em}| Base\_Exp + P1\_Exp} \\
                              & \text{\hspace{-16em}| Base\_Exp - P1\_Exp} \\
                              & \text{\hspace{-16em}...} \\

\\
\text{Base\_Exp}  \to & \\ 
                      & \text{\hspace{-16em}Literal} \\
                      & \text{\hspace{-16em}| Var} \\
                      & \text{\hspace{-16em}| ( Exp )} \\
                      & \text{\hspace{-16em}| Base\_Exp ( )} \\
                      & \text{\hspace{-16em}| Base\_Exp ( Args )} \\
                      & \text{\hspace{-16em}| - Base\_Exp} \\

\\
\text{Literal}  \to \text{INT | "true" | "false"} \\

\\
\text{Var}  \to \text{ID} \\

\\
\text{Args}  \to \text{Exp | Exp , Args} \\

\\
\text{Typed\_Id}  \to \text{Var : Type} \\

\\
\text{Params}  \to \text{Typed\_Id | Typed\_Id , Params} \\

\end{array}
\]
\end{minipage}
\caption{Our Subset of TypeScript. The start nonterminal is Statements, in the upper left.}\label{fig:typescript_grammar}
\end{figure}
\begin{figure}[ht]
\centering
\textbf{Inference Typing Rules for Expressions and Statements}
\begin{minipage}{0.4\textwidth}
\centering
\textbf{Expression Rules}
\begin{mathpar}
\inferrule[Int]
  {\;}{\Gamma \vdash \text{0} \implies \text{int}}
\end{mathpar}

\begin{mathpar}
\inferrule[True]
  {\;}{\Gamma \vdash \text{True} \implies \text{bool}}
\end{mathpar}

\begin{mathpar}
\inferrule[Var]
  {(x, \tau, \_) \in \Gamma}{\Gamma \vdash x \implies \tau}
\end{mathpar}

\begin{mathpar}
\inferrule[Sum]
  {\Gamma \vdash e1 \implies \text{int} \quad \Gamma \vdash e2 \implies \text{int}}
  {\Gamma \vdash e1 + e2 \implies \text{int}}
\end{mathpar}

\begin{mathpar}
\inferrule[Ternary Expression]
  {\Gamma \vdash e_1 \implies \text{bool} \quad \Gamma \vdash e_2 \implies \tau \quad \Gamma \vdash e_3 \implies \tau}{\Gamma \vdash (e_1 ? e_2 : e_3) \implies \tau}
\end{mathpar}

\begin{mathpar}
\inferrule[Function Application]
  {\Gamma \vdash f \implies \tau_1, \dots, \tau_n \to \tau \quad \forall j < n.~\Gamma \vdash x_j \implies \tau_j}{\Gamma \vdash f(x_1, \dots, x_n) \implies \tau}
\end{mathpar}
\end{minipage}
\caption{Selected Inference Rules for Typing Individual Expressions.} \label{fig:typescript_single_program_inference_rules_expressions}
\end{figure}
\begin{figure}[ht]
\centering
\textbf{Inference Typing Rules for Individual Statements}

\begin{minipage}{0.4\textwidth}
\centering
\textbf{Statement Rules}
\begin{mathpar}
\inferrule[Expression Statement]
  {\Gamma \vdash e \implies \tau' \quad \Gamma \vdash \bar{s} \implies \tau}{\Gamma \vdash e ; \bar{s} \implies \tau}
\end{mathpar}

\begin{mathpar}
\inferrule[Let Variable Declaration]
  {\Gamma \vdash e \implies \top \quad \Gamma + (x, \tau, \text{mutable}) \vdash \bar{s} \implies \tau'}{\Gamma \vdash \text{let} \ x : \tau = e ; \bar{s} \implies \tau'}
\end{mathpar}

\begin{mathpar}
\inferrule[Const Variable Declaration]
  {\Gamma \vdash e \implies \tau \quad \Gamma + (x, \tau, \text{immutable}) \vdash \bar{s} : \tau'}{\Gamma \vdash \text{const} \ x : \tau = e ; \bar{s} : \tau'}
\end{mathpar}

\begin{mathpar}
\inferrule[If-Then-Else -- May Not Return]
  {\Gamma \vdash e \implies \text{bool} \quad \Gamma \vdash s_1 \implies \tau \quad \Gamma \vdash s_2 \implies \tau \quad \Gamma \vdash \bar{s} \implies \tau}{\Gamma \vdash \text{if} \ e \ \text{then} \ s_1 \ \text{else} \ s_2 ; \bar{s} \implies \tau}
\end{mathpar}

\begin{mathpar}
\inferrule[If-Then-Else -- Definitely Returns]
  {\Gamma \vdash e \implies \text{bool} \quad \Gamma \vdash s_1 \implies \tau - () \quad \Gamma \vdash s_2 \implies \tau - () \quad \Gamma \vdash \bar{s} \implies \top}{\Gamma \vdash \text{if} \ e \ \text{then} \ s_1 \ \text{else} \ s_2 ; \bar{s} \implies \tau}
\end{mathpar}

\begin{mathpar}
\inferrule[While]
  {\Gamma \vdash e \implies \text{bool} \quad \Gamma \vdash s \implies \tau \quad \Gamma \vdash \bar{s} \implies \tau}{\Gamma \vdash \text{while} (e) \ s ; \bar{s} \implies \tau}
\end{mathpar}

\begin{mathpar}
\inferrule[No-op]
  {\;}{\Gamma \vdash \cdot \implies ()}
\end{mathpar}
\end{minipage}
\caption{Selected Typing Rules for Individual Statements. We give the typing rules of individual programs in our subset of TypeScript, eliding some trivial cases.
Note that $\bar{s}$ refers to a (possibly empty) sequence of statements.
We use $\cdot$ to denote the empty sequence of statements.
\changed{Note that statements type to their return types, as in~\cite{typescript_good_typesystem}.}} \label{fig:typescript_single_program_inference_rules_statements}
\end{figure}

\pagebreak
\clearpage
\subsubsection{Benchmarks and Additional Data}
\label{app:typescript_results}

\paragraph{Prompting}
All TypeScript experiments used the following context in instruct mode. The context is borrowed from \citet{typespaper} with the exception of the final 3 lines instructing the model to stay within our language subset.
\begin{verbatim}
You are an expert in typescript programming. 
Solve the given problem by writing solution code in typescript.
When answering, insert the solution code in a ```typescript...``` block.
Do NOT use arrays, strings, or lambdas!
Do NOT use console.log or add tests!
Do NOT throw errors!
\end{verbatim}

Each benchmark consists of a comment describing the task and the function signature, e.g.,

\begin{verbatim}
//Write a typescript function to identify non-prime numbers.
function is_not_prime(n: number): boolean 
\end{verbatim}

We use the comment as our prompt, and we use the signature as the prefix from which the LLM starts generating its solution, adding \"```typescript\" to the front. In our example, the prompt would be
\begin{verbatim}
//Write a typescript function to identify non-prime numbers.
\end{verbatim}

and the LLM would begin generation by appending to the string 

\begin{verbatim}
```typescript
function is_not_prime(n: number): boolean 
\end{verbatim}

\paragraph{Benchmarks}
We ran our experiments on the following 72 benchmarks from the MBPP~\cite{mbpp} benchmarks available in the MultiPL-E dataset~\cite{multiplE}:
\begin{verbatim}
mbpp_80_tetrahedral_number
mbpp_392_get_max_sum
mbpp_171_perimeter_pentagon
mbpp_127_multiply_int
mbpp_435_last_Digit
mbpp_287_square_Sum
mbpp_606_radian_degree
mbpp_803_is_perfect_square
mbpp_731_lateralsurface_cone
mbpp_581_surface_Area
mbpp_135_hexagonal_num
mbpp_17_square_perimeter
mbpp_77_is_Diff
mbpp_126_sum
mbpp_266_lateralsurface_cube
mbpp_797_sum_in_range
mbpp_3_is_not_prime
mbpp_458_rectangle_area
mbpp_441_surfacearea_cube
mbpp_162_sum_series
mbpp_448_cal_sum
mbpp_738_geometric_sum
mbpp_239_get_total_number_of_sequences
mbpp_59_is_octagonal
mbpp_638_wind_chill
mbpp_577_last_Digit_Factorial
mbpp_84_sequence
mbpp_724_power_base_sum
mbpp_641_is_nonagonal
mbpp_279_is_num_decagonal
mbpp_72_dif_Square
mbpp_781_count_divisors
mbpp_309_maximum
mbpp_14_find_Volume
mbpp_167_next_power_of_2
mbpp_600_is_Even
mbpp_742_area_tetrahedron
mbpp_432_median_trapezium
mbpp_234_volume_cube
mbpp_422_find_Average_Of_Cube
mbpp_292_find
mbpp_389_find_lucas
mbpp_227_min_of_three
mbpp_388_highest_Power_of_2
mbpp_271_even_Power_Sum
mbpp_67_bell_number
mbpp_274_even_binomial_Coeff_Sum
mbpp_86_centered_hexagonal_number
mbpp_574_surfacearea_cylinder
mbpp_430_parabola_directrix
mbpp_406_find_Parity
mbpp_605_prime_num
mbpp_264_dog_age
mbpp_770_odd_num_sum
mbpp_453_sumofFactors
mbpp_244_next_Perfect_Square
mbpp_93_power
mbpp_291_count_no_of_ways
mbpp_637_noprofit_noloss
mbpp_293_otherside_rightangle
mbpp_592_sum_Of_product
mbpp_256_count_Primes_nums
mbpp_479_first_Digit
mbpp_267_square_Sum
mbpp_58_opposite_Signs
mbpp_103_eulerian_num
mbpp_20_is_woodall
mbpp_96_divisor
mbpp_404_minimum
mbpp_752_jacobsthal_num
mbpp_765_is_polite
mbpp_801_test_three_equal
\end{verbatim}

\clearpage
\paragraph{Additional Results}
We include the total number of guesses required per token for each of our three generation modes in \Cref{fig:token_generation_appendix}.

\begin{figure}[t]
\centering
    \begin{subfigure}[b]{0.44\linewidth}
    \centering
        \includegraphics[width=0.97\linewidth]{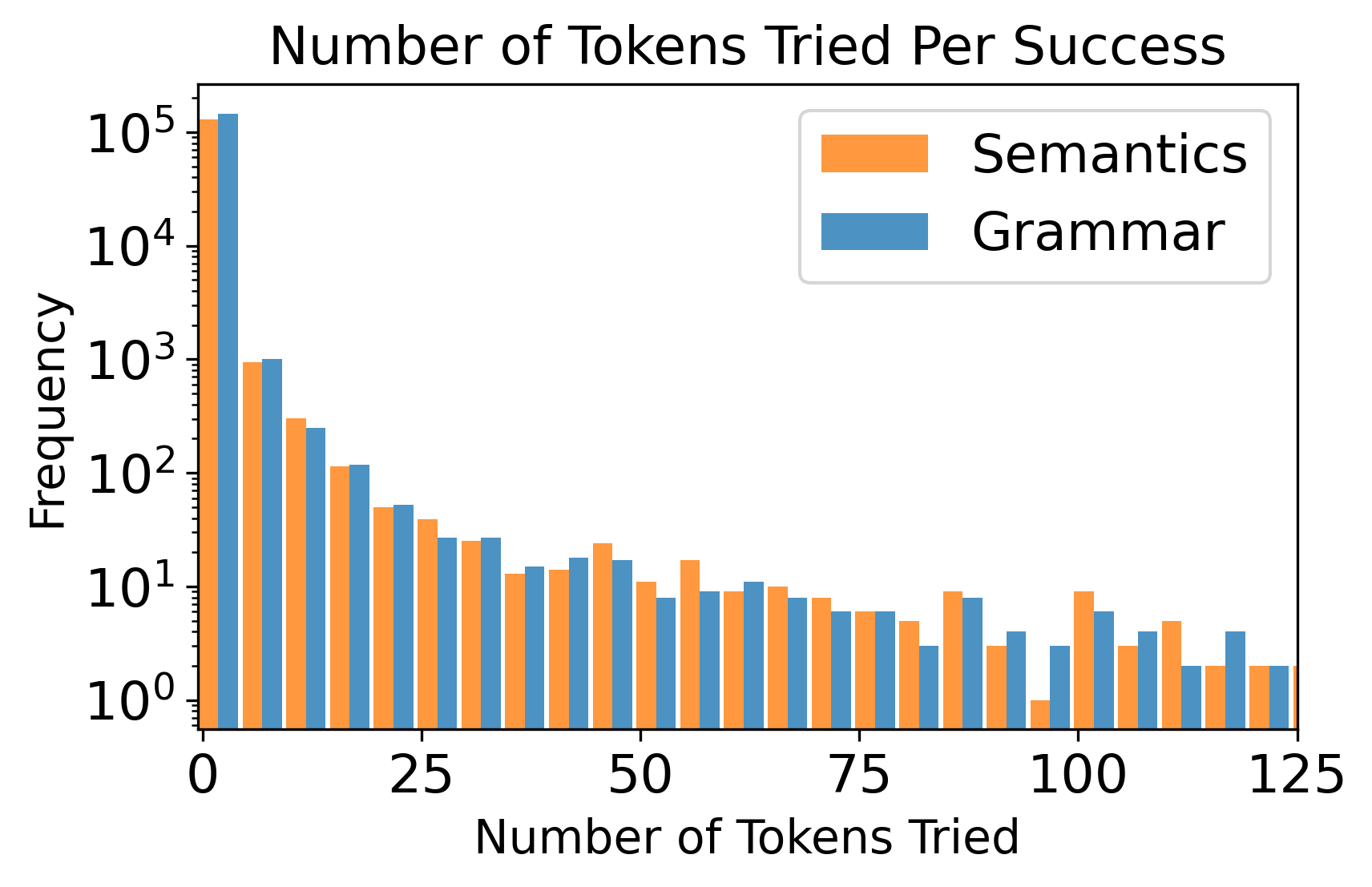}
        \caption{TypeScript, DeepSeek-Coder-6.7b, all temperatures}
    \end{subfigure}
    \hfill
    \begin{subfigure}[b]{0.44\linewidth}
    \centering
        \includegraphics[width=0.97\linewidth]{Images/tries_per_token_buckets_llama7b_typescript.png}
        \caption{TypeScript, CodeLlama7B, all temperatures}
    \end{subfigure}
    \hfill
    \begin{subfigure}[b]{0.44\linewidth}
    \centering
        \includegraphics[width=0.97\linewidth]{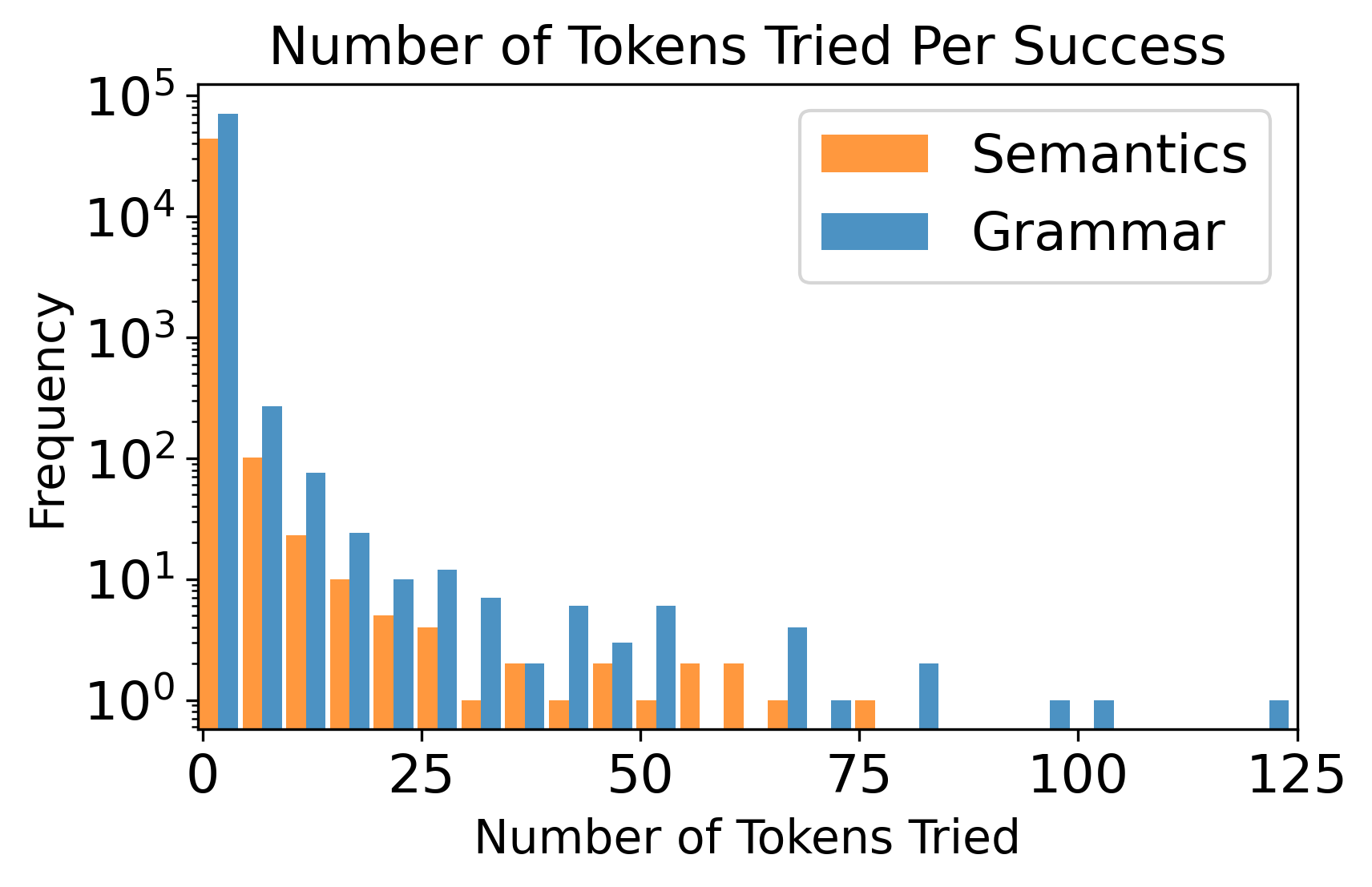}
        \caption{TypeScript, CodeLlama13B, all temperatures}
    \end{subfigure}
    \caption{Distribution of how many tokens were proven unrealizably by semantic constrained decoding to produce each individual token. The $k$th bar gives the number of successful tokens that were produced after trying between $5k$ and $5k+4$ unsuccesful tokens by CodeLlama-7b.}
    \label{fig:token_generation_appendix}
\end{figure}

\subsection{Equivalence-Guided Decoding}
\label{app:egraphs_results}
\subsubsection{Benchmarks and Additional Data.}
\ref{app:egraph-rules} shows a variation of our equivalence-guided decoding benchmarks where the rules were added to the prompt for the LM.

\begin{table}[t]
\centering
\caption{Equivalence-guided decoding benchmarks with rewrite rules added to prompts.}
\label{app:egraph-rules}
\resizebox{\textwidth}{!}{%
\begin{tabular}{@{}ll|ccccc:c|ccccc:c|ccccc:c@{}}
\toprule
 & & \multicolumn{6}{c|}{\textbf{DeepSeek-Coder-6.7b}} & \multicolumn{6}{c|}{\textbf{CodeLlama-7B}} & \multicolumn{6}{c}{\textbf{CodeLlama-13B}} \\
 & & \multicolumn{5}{c:}{\textbf{Temperature}} & \multirow{2}{*}{\textbf{Tot.}} & \multicolumn{5}{c:}{\textbf{Temperature}} & \multirow{2}{*}{\textbf{Tot.}}  & \multicolumn{5}{c:}{\textbf{Temperature}} & \multirow{2}{*}{\textbf{Tot.}}  \\
 & & 0.01 & 0.3 & 0.5 & 0.7 & 1.0 & & 0.01 & 0.3 & 0.5 & 0.7 & 1.0 & & 0.01 & 0.3 & 0.5 & 0.7 & 1.0 & \\
\midrule[1pt]
\multirow{3}{*}{\shortstack{\textbf{Equivalence}\\\textbf{Rules + No Delimit}\\\textbf{(10 programs)}}}
 & Unconstrained     & 0 & 0 & 0 & 0 & 0 & 0 & 0 & 0 & 0 & 0 & 0 & 0 & 0 & 0 & 0 & 0 & 0 & 0 \\
 & Grammar        & 0 & 0 & 0 & 0 & 0 & 0 & 0 & 0 & 0 & 0 & 0 & 0 & 0 & 0 & 0 & 0 & 0 & 0 \\
 & Semantic  & \textbf{4} & \textbf{3} & \textbf{4} & \textbf{6} & \textbf{3} & \textbf{20} & \textbf{8} & \textbf{6} & \textbf{10} & \textbf{6} & \textbf{4} & \textbf{34} & \textbf{9} & \textbf{10} & \textbf{10} & \textbf{7} & \textbf{6} & \textbf{42} \\[2pt]
\cdashline{1-20}
\multirow{3}{*}{\shortstack{\textbf{Equivalence}\\\textbf{Rules + Delimit}\\\textbf{(10 programs)}}}
 & Unconstrained     & 0 & 0 & 0 & 0 & 0 & 0 & 3 & 3 & 5 & 0 & 1 & 12 & 2 & 2 & 4 & 2 & 0 & 10 \\
 & Grammar        & 0 & 0 & 0 & 0 & 0 & 0 & 3 & 3 & 5 & 4 & 1 & 16 & 2 & 6 & 1 & 2 & 3 & 14 \\
 & Semantic  & \textbf{8} & \textbf{8} & \textbf{7} & \textbf{8} & \textbf{6} & \textbf{37} & \textbf{8} & \textbf{7} & \textbf{5} & \textbf{8} & \textbf{8} & \textbf{36} & \textbf{8} & \textbf{8} & \textbf{7} & \textbf{9} & \textbf{7} & \textbf{39} \\
\bottomrule
\end{tabular}
}
\end{table}

\paragraph{Context.} The equivalence benchmarks use the following context. The last line is removed for the NO-DELIMIT experiments.
\begin{verbatim}
You are a code refactoring assistant for a simple functional language. 
The language consists of expressions which are either identifiers, integers, 
basic arithmetic operations, function application, and let expressions. 
The only binary operators are +, -, *, and /. 
All other functions (for example, sqrt or pow) are named---
ONLY use names appearing in the original program.

As examples, syntactically valid programs would include:

```
let x = sqrt 42 in
let y = pow (f x) 2 in
y - 3
```

and

```
f x + g y
```

Your job is to refactor programs into *equivalent* ones which also
have clear, readable style using let bindings when helpful. 
Never introduce new features not in the language. 
Never include comments or explanations. 
ONLY output code, then IMMEDIATELY stop. 
Never redefine variables in the original program 
or that have already been defined.

Start and end your solution with a codeblock using ```.
\end{verbatim}

\paragraph{Benchmarks.} We show the 10 benchmark programs we used below.

\begin{verbatim}
1.  fetch_document (authorize_user_for_document (
    authenticate_user current_user web_request) document_id)

2.  sqrt (pow (x1 - x2) 2 + pow (y1 - y2) 2)

3.  pow 10 (-15) * (66743 * m1 * m2) / (pow r 2)

4.  add_watermark (apply_filter (
   crop_image original_image selection) filter_type) watermark_image

5.  start + (end - start) * scale

6.  (sum (filter positive xs)) / (length (filter positive xs))

7.  power / 1000 * hours * price_per_kwh

8.  (-b + sqrt ((pow b 2) - 4 * a * c)) / (2 * a)

9.  map toUpper (filter isAlpha s)

10. sqrt ((pow (a - ((a+b+c)/3)) 2) + 
    (pow (b - ((a+b+c)/3)) 2) + (pow (c - ((a+b+c)/3)) 2)) / 3
\end{verbatim}

\paragraph{Egglog file.} We show the Egglog file defining the rewrites for the initial e-graph below. It encodes basic arithmetic rules.

\begin{verbatim}
(datatype Math
  (Num i64)
  (Str String)
  (Var String)
  (Add Math Math)
  (Sub Math Math)
  (Neg Math)
  (Pow Math Math)
  (Sqrt Math)
  (Mul Math Math)
  (Div Math Math)
  (App Math Math))

(rewrite (Add a b)
         (Add b a))

(rewrite (Add (Num a) (Num b))
         (Num (+ a b)))

(rewrite (Add (Add a b) c)
         (Add a (Add b c)))

(rewrite (Neg a)
         (Sub (Num 0) a))

(rewrite (Sub (Num 0) a)
         (Neg a))

(rewrite (Sub a b)
         (Add a (Mul (Num -1) b)))

(rewrite (Sub (Num a) (Num b))
         (Num (- a b)))

(rewrite (Mul a b)
         (Mul b a))

(rewrite (Mul (Num a) (Num b))
         (Num (* a b)))

(rewrite (Mul (Mul a b) c)
         (Mul a (Mul b c)))

(rewrite (Mul a (Add b c))
         (Add (Mul a b) (Mul a c)))

(rewrite (Div a b)
         (Mul a (Div (Num 1) b)))

(rewrite (Mul a (Div (Num 1) b))
         (Div a b))

(rewrite (Div (Num 1) (Mul b c))
         (Mul (Div (Num 1) b) (Div (Num 1) c)))

(rewrite (Mul (Div (Num 1) b) (Div (Num 1) c))
         (Div (Num 1) (Mul b c)))
\end{verbatim}

\end{document}